\documentclass[10pt,a4paper,utf8]{article}

\usepackage[a4paper,inner=100pt,outer=100pt,textheight=630pt,centering]{geometry}

\usepackage{enumerate}
\usepackage[usenames,dvipsnames]{xcolor}
\usepackage[pagebackref,citecolor=blue,linkcolor=OliveGreen,urlcolor=Mahogany,colorlinks]{hyperref}
\usepackage[utf8]{inputenc}
\usepackage{color}
\usepackage{url}
\usepackage{amssymb}
\usepackage{amsmath}
\usepackage{amsthm}
\usepackage{cases}
\usepackage[innertopmargin=-5,skipbelow=-5]{mdframed}
\usepackage{stmaryrd}
\usepackage{mathpartir}
\usepackage{tikz}
\usepackage{tikz-cd}
\usetikzlibrary{arrows,matrix,decorations.pathmorphing,
  decorations.markings, calc, backgrounds}

\theoremstyle{plain}
\newtheorem{theorem}{Theorem}
\newtheorem*{theorem*}{Theorem}
\newtheorem{lemma}[theorem]{Lemma}
\newtheorem*{lemma*}{Lemma}
\newtheorem{corollary}[theorem]{Corollary}
\newtheorem*{corollary*}{Corollary}

\newtheorem*{proposition*}{Proposition}

\theoremstyle{definition}
\newtheorem{definition}[theorem]{Definition}
\newtheorem*{definition*}{Definition}
\newtheorem{example}[theorem]{Example}
\newtheorem*{example*}{Example}

\theoremstyle{remark}

\newtheorem*{remark*}{Remark}

\newcommand{\Haskell}{{\sc Haskell}}

\newcommand{\eg}{{e.g.}}
\newcommand{\ie}{{i.e.}}

\newcommand{\sect}[1]{Section~\ref{#1}}
\newcommand{\fig}[1]{Figure~\ref{#1}}

\renewcommand{\|}{{~\mid~}}

\renewcommand{\=}{{~~ = ~~}}

\newcommand{\der}{\vdash}
\newcommand{\emptyctxt}{()}
\newcommand{\inctxt}[2]{#1\der#2}

\newcommand{\set}[1]{\{ #1 \}}

\newcommand{\subst}[2]{(#1/#2)}
\renewcommand{\mod}[1]{\; ({\sf mod.} \; #1)}
\newcommand{\hmapsto}{~}

\newcommand\myeq{f}

\newcommand{\pabs}[1]{{\langle}#1{\rangle} \; }

\newcommand{\refl}[1]{{1_{#1}}}   % {{\sf refl_{#1}}}
\DeclareMathOperator{\transport}{\mathsf{transp}} % transport as special case of comp
\newcommand{\J}{\mathsf{J}}
\DeclareMathOperator{\idrefl}{\mathsf{r}}

\newcommand{\dM}{{\sf dM}}
\newcommand{\II}{\mathbb{I}}
\newcommand{\FF}{\mathbb{F}}

\newcommand{\NN}{{\sf N}}
\newcommand{\natrec}{{\sf natrec}}
\newcommand{\s}{{\sf s}}
\newcommand{\suc}[1]{\s \; #1}
\newcommand{\0}{{\sf 0}}

\newcommand{\Id}{{\sf Id}}
\newcommand{\ident}{{\sf id}}
\newcommand{\Path}{{\sf Path}}
\newcommand{\Equiv}{{\sf Equiv}}
\newcommand{\isEquiv}{{\sf isEquiv}}
\newcommand{\ext}{{\sf contr}}
\newcommand{\isContr}{{\sf isContr}}

\DeclareMathOperator{\dom}{dom}

\def\NN{\hbox{\sf N}}

\def\N0{\hbox{\sf N}_0}

\newcommand{\lam}[2]{{\langle}#1{\rangle}#2}
\newcommand{\id}{\mathrm{id}}
\newcommand{\pp}{{\sf p}}
\newcommand{\qq}{{\sf q}}
\newcommand{\comp}{{\sf comp}}
\newcommand{\hcomp}{{\sf hcomp}}
\newcommand{\pres}[1]{{\sf pres}^{#1}}

\newcommand{\eq}{{\sf equiv}}

\newcommand{\Comp}{{\sf fill}}

\newcommand{\base}{{\sf base}}
\newcommand{\inc}{{\sf inc}}
\newcommand{\inh}{\mathsf{inh}}
\newcommand{\squash}{{\sf squash}}
\newcommand{\transp}{{\sf transp}} % defined operation on HITs
\newcommand{\squeeze}{{\sf squeeze}}
\newcommand{\LOOP}{{\sf loop}}
\newcommand{\Sp}{{\sf S}}

\newcommand{\Glue}{{\sf Glue}}
\newcommand{\glue}{{\sf glue}}
\newcommand{\ugl}{{\sf unglue}}
\newcommand{\U}{\mathsf{U}}
\DeclareMathOperator{\eqToPath}{\mathsf{eqToPath}}
\newcommand{\can}{\mathsf{pathToEq}} % "the" canonical map for ua

\newcommand{\Set}{\mathbf{Set}}
\newcommand{\Top}{\mathbf{Top}}
\DeclareMathOperator{\sing}{\mathrm{S}} 
\newcommand{\CC}{\mathcal{C}}
\renewcommand{\deg}{\mathrm{s}}
\newcommand{\op}{\mathrm{op}}
\DeclareMathOperator{\Hom}{Hom}
\newcommand{\cwf}{CwF}
\DeclareMathOperator{\Ty}{Ty}
\DeclareMathOperator{\FTy}{FTy} % for "fibrant" types
\DeclareMathOperator{\Ter}{Ter}
\newcommand{\den}[1]{[\![#1]\!]} % denotation
\DeclareMathOperator{\yoneda}{\mathbf{y}}
\DeclareMathOperator{\scomp}{\mathtt{comp}} % semantic composition
\DeclareMathOperator{\sGlue}{\mathtt{Glue}} % semantic Glue
\DeclareMathOperator{\sglue}{\mathtt{glue}} % semantic glue
\DeclareMathOperator{\sugl}{\mathtt{unglue}} % semantic unglue
\newcommand{\sNN}{\mathtt{N}}

\DeclareMathOperator{\sequiv}{\mathtt{e}} % equivalence structure

\newcommand{\sPath}{\mathtt{Path}}
\newcommand{\spabs}{\langle \, \rangle}

\newcommand{\NAT}{\mathbb{N}}   % {0,1,2,..}
\newcommand{\UU}{\mathtt{U}}   % semantic universe

\DeclareMathOperator{\app}{\mathtt{app}}

\DeclareMathOperator{\El}{\mathsf{El}}
\newcommand{\code}[1]{\ulcorner #1 \urcorner}

\tikzset{
    comp line/.style={thick,draw=black,decorate, decoration={snake},
      ->, >=stealth'},
    path line/.style={thick,draw=black},
    fun line/.style={thick,draw=black, ->, >=stealth'}
}

\tikzset{
    comp line/.style={thick,draw=black,decorate, decoration={snake},
      ->, >=stealth'},
    path line/.style={thick,draw=black},
    fun line/.style={thick,draw=black, ->, >=stealth'}
}

\begin{document}

\title{Cubical Type Theory: a constructive interpretation of the
  univalence axiom\footnote{This material is based upon work supported
    by the National Science Foundation under agreement
    No. DMS-1128155. Any opinions, findings and conclusions or
    recommendations expressed in this material are those of the
    author(s) and do not necessarily reflect the views of the National
    Science Foundation.}}

\author{Cyril Cohen, Thierry Coquand, Simon Huber, and Anders Mörtberg}

\date{}

\maketitle

\begin{abstract}
  This paper presents a type theory in which it is possible to
  directly manipulate $n$-dimensional cubes (points, lines, squares,
  cubes, etc.) based on an interpretation of dependent type theory in
  a cubical set model. This enables new ways to reason about identity
  types, for instance, function extensionality is directly provable in
  the system. Further, Voevodsky's univalence axiom is provable in
  this system. We also explain an extension with some higher inductive
  types like the circle and propositional truncation. Finally we
  provide semantics for this cubical type theory in a constructive
  meta-theory.
\end{abstract}

\tableofcontents

\section{Introduction}\label{sec:introduction}

This work is a continuation of the program started in
\cite{BCH,simonlic} to provide a constructive justification of
Voevodsky's univalence axiom \cite{Voevodsky}. This axiom allows many
improvements for the formalization of mathematics in type theory:
function extensionality, identification of isomorphic structures,
etc. In order to preserve the good computational properties of type
theory it is crucial that postulated constants have a computational
interpretation. Like in \cite{BCH,simonlic,Pitts} our work is based on
a nominal extension of $\lambda$-calculus, using {\em names} to
represent formally elements of the unit interval $[0,1]$. This paper
presents two main contributions.

The first one is a refinement of the semantics presented in
\cite{BCH,simonlic}.  We add new operations on names corresponding to
the fact that the interval $[0,1]$ is canonically a de~Morgan
algebra~\cite{Balbes}.  This allows us to significantly simplify our
semantical justifications. In the previous work, we noticed that it is
crucial for the semantics of higher inductive types~\cite{hott-book}
to have a ``diagonal'' operation. By adding this operation we can
provide a semantical justification of some higher inductive types and
we give two examples (the spheres and propositional
truncation). Another shortcoming of the previous work was that using
path types as equality types did not provide a justification of the
computation rule of the Martin-L\"of identity type~\cite{ML75} as a
judgmental equality. This problem has been solved by
Andrew Swan~\cite{Swan}, in the framework of \cite{BCH,simonlic,Pitts},
who showed that we can define a new type, {\em equivalent to}, but not
judgmentally equal to the path type. This has a simple definition in
the present framework.

The second contribution is the design of a type system\footnote{We
  have implemented a type-checker for this system in \Haskell{}, which
  is available at:\\ \url{https://github.com/mortberg/cubicaltt}}
inspired by this semantics which extends Martin-L\"of type theory
\cite{MLTT72,ML75}. We add two new operations on contexts: addition of
new names representing dimensions and a restriction operation. Using
these we can define a notion of extensibility which generalizes the
notion of being connected by a path, and then a Kan composition
operation that expresses that being extensible is preserved along
paths.  We also define a new operation on types which expresses that
this notion of extensibility is preserved by equivalences. The axiom
of univalence, and composition for the universe, are then both
expressible using this new operation.

\medskip

The paper is organized as follows. The first part,
Sections~\ref{sec:basictt} to \ref{sec:universe}, presents the type
system. The second part, \sect{sec:semantics}, provides its semantics
in cubical sets. Finally, in \sect{sec:extensions}, we present two
possible extensions: the addition of an identity type, and two
examples of higher inductive types.

\section{Basic type theory}\label{sec:basictt}
% - Typing rules (pi,sigma,U)
% - Equality rules: beta, eta, surjective pairing
% - Ordinary substitution (x=u)

In this section we introduce the version of dependent type theory on
which the rest of the paper is based. This presentation is standard,
but included for completeness. The type theory that we consider has a
type of natural numbers, but no universes (we consider the addition of
universes in~\sect{sec:universe}). It also has $\beta$ and
$\eta$-conversion for dependent functions and surjective pairing for
dependent pairs.

The syntax of contexts, terms and types is specified by:
\[
\begin{array}{lcll}
  \Gamma,\Delta & ::= & \emptyctxt \| \Gamma, x : A & \hspace{2cm}\text{Contexts}\\ \\
  t,u,A,B       & ::= & x \| \lambda x : A. \; t \| t \; u \| (x : A) \rightarrow B & \hspace{2cm}\Pi\text{-types} \\
                & \|  & (t,u) \| t.1 \| t.2 \| (x : A) \times B & \hspace{2cm}\Sigma\text{-types} \\
                & \|  & \0 \| \suc{u} \| \natrec~t~u \| \NN & \hspace{2cm}\text{Natural numbers} \\
\end{array}
\]

We write $A \rightarrow B$ for the non-dependent function space and
$A \times B$ for the type of non-dependent pairs. Terms and types are
considered up to $\alpha$-equivalence of bound
variables. Substitutions, written
$\sigma = ({x_1}/{u_1}, \dots, {x_n}/{u_n})$, are defined to act on
expressions as usual, i.e., simultaneously replacing $x_i$ by $u_i$,
renaming bound variables whenever necessary.  The inference rules of
this system are presented in~\fig{basictt:rules} where in the
$\eta$-rule for $\Pi$- and $\Sigma$-types we omitted the premises that
$t$ and $u$ should have the respective type.

We define $\Delta \der \sigma:\Gamma$ by induction on $\Gamma$. We
have $\Delta \der ():()$ (empty substitution) and $\Delta \der
(\sigma,{x}/{u}):\Gamma,x:A$ if $\Delta\der \sigma :\Gamma$ and
$\Delta \der u:A \sigma$.

We write $\mathsf{J}$ for an arbitrary judgment and, as usual, we
consider also {\em hypothetical} judgments $\Gamma\der\mathsf{J}$ in a
{\em context} $\Gamma$.

\begin{figure}
\fbox{Well-formed contexts, $\Gamma \der$} (The condition
$x \notin \dom(\Gamma)$ means that $x$ is not declared in $\Gamma$)
\begin{mathpar}
  \inferrule { } {\emptyctxt \der {}} %
  \and %
  \inferrule {\Gamma \der A} {\Gamma, x : A \der {}}
             \quad{(x \notin \dom(\Gamma))} %
\end{mathpar}
\fbox{Well-formed types, $\Gamma \der A$}
\begin{mathpar}
  \inferrule {\Gamma, x : A \der B} {\Gamma \der (x : A) \rightarrow B} %
  \and %
  \inferrule {\Gamma, x : A \der B} {\Gamma \der (x : A) \times B} %
  \and %
  \inferrule {\Gamma \der} {\Gamma \der \NN} %
\end{mathpar}
\fbox{Well-typed terms, $\Gamma \der t : A$}
\begin{mathpar}
  \inferrule {\Gamma \der t : A \\ \Gamma \der A = B} {\Gamma \der t : B} %
  \and %
  \inferrule {\Gamma, x : A \der t : B}
             {\Gamma \der \lambda x : A. \; t : (x : A) \rightarrow B} %
  \and %
  \inferrule {\Gamma \der} {\Gamma \der x : A}~{(x : A \in \Gamma)} %
  \and %
  \inferrule {\Gamma \der t : (x : A) \rightarrow B \\ \Gamma \der u : A}
             {\Gamma \der t \; u : B \subst{x}{u}} %
  \and %
  \inferrule {\Gamma \der t : (x : A) \times B}
             {\Gamma \der t.1 : A} %
  \and %
  \inferrule {\Gamma \der t : (x : A) \times B}
             {\Gamma \der t.2 : B \subst{x}{t.1}} %
  \and %
  \inferrule {\Gamma, x : A \der B \\ \Gamma \der t : A \\
              \Gamma \der u : B \subst{x}{t}}
             {\Gamma \der (t,u) : (x : A) \times B} %
  \and %
  \inferrule { \Gamma \der {}} {\Gamma \der \0 : \NN} %
  \and %
  \inferrule {\Gamma \der n : \NN} {\Gamma \der \suc{n} : \NN} %
  \and %
  \inferrule {\Gamma, x : \NN \der P \\ \Gamma \der a : P\subst{x}{\0} \\
              \Gamma \der b : (n : \NN) \rightarrow P\subst{x}{n}
              \rightarrow P\subst{x}{\suc{n}}}
             {\Gamma \der \natrec~a~b : (x : \NN) \rightarrow P} %
\end{mathpar}
\fbox{Type equality, $\Gamma \der A = B$} (Congruence and equivalence
rules which are omitted)

\vspace{0.07cm}
\fbox{Term equality, $\Gamma \der a = b : A$} (Congruence and
equivalence rules are omitted)
\begin{mathpar}
  \inferrule {\Gamma \der t = u : A \\ \Gamma \der A = B}
             {\Gamma \der t = u : B} %
  \and %
  \inferrule {\Gamma, x : A \der t : B \\ \Gamma \der u : A}
             {\Gamma \der (\lambda x : A. \; t) \; u =
              t \subst{x}{u} : B \subst{x}{u}} %
  \and %
  \inferrule {\Gamma,x:A \der t~x = u~x: B}
              % \Gamma \der t : (x : A) \rightarrow B\\
              % \Gamma \der u : (x : A) \rightarrow B}
             {\Gamma \der t = u : (x : A) \rightarrow B} %
  \and %
  \inferrule {\Gamma, x : A \der B \\
              \Gamma \der t : A \\ \Gamma \der u : B \subst{x}{t}}
             {\Gamma \der (t,u).1 = t : A} %
  \and %
  \inferrule {\Gamma, x : A \der B \\
              \Gamma \der t : A \\ \Gamma \der u : B \subst{x}{t}}
             {\Gamma \der (t,u).2 = u : B \subst{x}{t}} %
  \and %
  \inferrule {\Gamma, x : A \der B \\
              \Gamma \der t.1 = u.1 : A \\ \Gamma \der t.2 = u.2 : B \subst{x}{t.1}}
             {\Gamma \der t = u : (x : A) \times B} %
  \and %
  \inferrule {\Gamma, x : \NN \der P \\ \Gamma \der a : P\subst{x}{\0} \\
              \Gamma \der b : (n : \NN) \rightarrow P\subst{x}{n}
              \rightarrow P\subst{x}{\suc{n}}}
             {\Gamma \der \natrec~a~b~\0 = a : P\subst{x}{\0}} %
  \and %
  \inferrule {\Gamma, x : \NN \der P \\ \Gamma \der a : P\subst{x}{\0} \\
              \Gamma \der b : (n : \NN) \rightarrow P\subst{x}{n}
              \rightarrow P\subst{x}{\suc{n}} \\ \Gamma \der n : \NN}
             {\Gamma \der \natrec~a~b~(\suc{n}) = b \; n \; (\natrec~a~b~n) : P\subst{x}{\suc{n}}} %
\end{mathpar}
\caption{Inference rules of the basic type theory}\label{basictt:rules}
\end{figure}

The following lemma will be valid for all extensions of type theory we
consider below.

\begin{lemma}
  Substitution is admissible:
  \begin{mathpar}
    \inferrule {\Gamma \der \mathsf{J} \\ \Delta \der \sigma : \Gamma}
               {\Delta \der \mathsf{J}\sigma} %
  \end{mathpar}
  In particular, weakening is admissible, i.e., a judgment valid in a
  context stays valid in any extension of this context.
\end{lemma}

\section{Path types}\label{sec:pathtypes}
% - Interval, II, de Morgan algebra
% - Name substitution (i=r), example of what "p @ -i", "p @ i /\ j", etc. means
% - Typing rules
% - Equality rules: p @ 0/1 = a/b, <i> p @ i = p
% - Examples: funext, contr singl

As in~\cite{BCH,Pitts} we assume that we are given a discrete infinite
set of names (representing directions) $i,j,k,\dots$ We define $\II$
to be the free de~Morgan algebra~\cite{Balbes} on this set of
names. This means that $\II$ is a bounded distributive lattice with
top element $1$ and bottom element $0$ with an involution $1 - r$
satisfying:
\begin{mathpar}
1 - 0 = 1 \and %
1 - 1 = 0 \and %
1 - (r \lor s) = (1 - r) \land (1 - s) \and %
1 - (r \land s) = (1 - r) \lor (1 - s) \and %
\end{mathpar}
The elements of $\II$ can hence be described by the following grammar:
\[
\begin{array}{lcl}
  r,s & ::= & 0 \| 1 \| i \| 1 - r \| r \wedge s \| r\vee s
\end{array}
\]
The set $\II$ also has decidable equality, and as a distributive
lattice, it can be described as the free distributive lattice
generated by symbols $i$ and $1-i$~\cite{Balbes}. As in~\cite{BCH},
the elements in $\II$ can be thought as formal representations of
elements in $[0,1]$, with $r \wedge s$ representing $min(r,s)$ and $r
\vee s$ representing $max(r,s)$. With this in mind it is clear that
$(1 - r) \land r \neq 0$ and $(1 -r) \lor r \neq 1$ in general.

\begin{remark*}
  We could instead also use a so-called Kleene
  algebra~\cite{Kalman58}, i.e., a de~Morgan algebra satisfying in
  addition $r \land (1-r) \leqslant s \lor (1-s)$. The free Kleene
  algebra on the set of names can be described as above but by
  additionally imposing the equations $i \land (1-i) \leqslant j \lor
  (1-j)$ on the generators; this still has a decidable equality.  Note
  that $[0,1]$ with the operations described above is a Kleene
  algebra.  With this added condition, $r = s$ if, and only if, their
  interpretations in $[0,1]$ are equal. A consequence of using a
  Kleene algebra instead would be that more terms would be
  judgmentally equal in the type theory.
\end{remark*}

\subsection{Syntax and inference rules}

Contexts can now be extended with name declarations:
\[
\begin{array}{lcl}
  \Gamma,\Delta & ::= & \dots \| \Gamma, i : \II \\
\end{array}
\]
together with the context rule:
\begin{mathpar}
  \inferrule {\Gamma \der} {\Gamma, i : \II \der}~{(i \notin \dom(\Gamma))} %
\end{mathpar}
A judgment of the form $\Gamma \der r : \II$ means that $\Gamma \der$
and $r$ in $\II$ depends only on the names declared in
$\Gamma$. The judgment $\Gamma \der r = s : \II$ means that $r$ and
$s$ are equal as elements of $\II$, $\Gamma \der r : \II$, and $\Gamma
\der s : \II$.  Note, that judgmental equality for $\II$ will be
re-defined once we introduce restricted contexts in
\sect{sec:compositions}.

The extension to the syntax of basic dependent type theory is:
\[
\begin{array}{lcll}
%   \Gamma,\Delta & ::= & \dots \| \Gamma, i : \II           & \text{Contexts}\\ \\
  t,u,A,B       & ::= & \dots & \\
                & \|  & \Path~A~t~u \| \pabs{i} t \| t \; r & \hspace{2cm} \text{Path types} \\
                %% & \|  & \pabs{i} t & ~~~~~ \text{Path abstraction} \\
                %% & \|  & t \; r    & \text{Path application} \\ \\
\end{array}
\]
Path abstraction, $\pabs{i} t$, binds the name $i$ in $t$, and path
application, $t \; r$, applies a term $t$ to an element $r : \II$. %
This is similar to the notion of name-abstraction in nominal
sets~\cite{PittsBook}.

The substitution operation now has to be extended to substitutions of
the form $\subst{i}{r}$. There are special substitutions of the form
$\subst{i}{0}$ and $\subst{i}{1}$ corresponding to taking faces of an
$n$-dimensional cube, we write these simply as $(i0)$ and $(i1)$.

The inference rules for path types are presented
in~\fig{pathtypes:rules} where again in the $\eta$-rule we omitted
that $t$ and $u$ should be appropriately typed.

\begin{figure}[h]
\begin{mdframed}
\begin{mathpar}
  \inferrule {\Gamma \der A \\ \Gamma \der t : A \\ \Gamma \der u : A}
             {\Gamma \der \Path~A~t~u} %
  \and %
  \inferrule {\Gamma \der A \\ \Gamma, i : \II \der t : A}
             {\Gamma \der \pabs{i} t : \Path~A~t(i0)~t(i1)} %
  \and %
  \inferrule {\Gamma \der t : \Path~A~u_0~u_1 \\ \Gamma \der r : \II}
             {\Gamma \der t \; r : A} %
  \and %
  \inferrule {\Gamma\der A \\ \Gamma, i : \II \der t : A \\
              \Gamma \der r : \II} %
             {\Gamma \der (\pabs{i} t) \; r = t \subst{i}{r} : A} %
  \and %
  \inferrule {\Gamma, i : \II \der t~i = u~i : A}
             {\Gamma \der t = u : \Path~A~u_0~u_1} %
  \and %
  \inferrule {\Gamma \der t : \Path~A~u_0~u_1} {\Gamma \der t \; 0 = u_0 : A} %
  \and %
  \inferrule {\Gamma \der t : \Path~A~u_0~u_1} {\Gamma \der t \; 1 = u_1 : A} %
\end{mathpar}
\end{mdframed}
\caption{Inference rules for path types}\label{pathtypes:rules}
\end{figure}

We define $\refl{a} : \Path~A~a~a$ as $\refl{a} = \pabs{i} a$, which
corresponds to a proof of reflexivity.

The intuition is that a type in a context with $n$ names corresponds
to an $n$-dimensional cube:
%
% % Old version implemented using tikzcd
% \begin{mathpar}
% \begin{array}{|c|c|}
%   \hline
%   \emptyctxt \der A & \begin{tikzcd} \bullet \; A \end{tikzcd} \\ \hline
%   i : \II \der A & \begin{tikzcd} A(i0) \arrow{r}{A} & A(i1) \end{tikzcd} \\ \hline
%   i : \II, j : \II \der A &
%        \begin{tikzcd}[row sep=1.25cm,column sep=1.25cm]
%             A(i0)(j1) \arrow{r}{A(j1)} & A(i1)(j1)  \\
%             A(i0)(j0) \arrow[phantom]{ur}{A} \arrow[swap]{r}{A(j0)} \arrow{u}{A(i0)} & A(i1)(j0) \arrow[swap]{u}{A(i1)}
%        \end{tikzcd} \\ \hline
%   \vdots & \vdots \\ \hline
% \end{array}
% \end{mathpar}
%
% % New version implemented using only tikz
\begin{center}
\setlength{\tabcolsep}{5mm} % separator between columns
\begin{tabular}{|c|c|}
  \hline
  $\emptyctxt \der A$ &
  $\bullet \; A$
  \\ \hline
  $i : \II \der A$ &
  \begin{tikzpicture}[baseline=-2,scale=2]
    \node (A0) at (0,0) {$A(i0)$};%
    \node (A1) at (1,0) {$A(i1)$};%
    \path[->,font=\scriptsize,>=angle 90]%
       (A0) edge node[above]{$A$} (A1);%
  \end{tikzpicture}
  \\ \hline
  $i : \II, j : \II \der A$ &
  \begin{tikzpicture}[baseline=30,xscale=3,yscale=2]
    \node (A01) at (0,1) {$A(i0)(j1)$};
    \node (A11) at (1,1) {$A(i1)(j1)$};
    \node (A00) at (0,0) {$A(i0)(j0)$};
    \node (A10) at (1,0) {$A(i1)(j0)$};
    \path[->,font=\scriptsize,>=angle 90]
      (A01) edge node[above]{$A(j1)$} (A11)
      (A00) edge node[left]{$A(i0)$} (A01)
      (A10) edge node[right]{$A(i1)$} (A11)
      (A00) edge node[below]{$A(j0)$} (A10);
  \end{tikzpicture}
  \\ \hline
  $\vdots$ &
  $\vdots$
  \\ \hline
\end{tabular}
\end{center}

Note that $A(i0)(j0) = A(j0)(i0)$. The substitution~$\subst{i}{j}$
corresponds to renaming a dimension, while $\subst{i}{1 - i}$
corresponds to the inversion of a path. If we have $i : \II \der p$
with $p(i0) = a$ and $p(i1) = b$ then it can be seen as a line
% % Old version:
% $$
% \begin{tikzcd}[column sep=2cm]
%   a \arrow{r}{p} & b
% \end{tikzcd}
% $$
\begin{center}
  \begin{tikzpicture}[scale=2.5]
    \node (a) at (0,0) {$a$};%
    \node (b) at (1,0) {$b$};%
    \path[->,font=\scriptsize,>=angle 90]%
       (a) edge node[above]{$p$} (b);%
  \end{tikzpicture}
\end{center}
in direction $i$, then:
% % Old version:
% $$
% \begin{tikzcd}[column sep=2cm]
%   b \arrow{r}{p \subst{i}{1-i}} & a
% \end{tikzcd}
% $$
\begin{center}
  \begin{tikzpicture}[scale=2.5]
    \node (b) at (0,0) {$b$};%
    \node (a) at (1,0) {$a$};%
    \path[->,font=\scriptsize,>=angle 90]%
       (b) edge node[above]{$p \subst{i}{1-i}$} (a);%
  \end{tikzpicture}
\end{center}

The substitutions $\subst{i}{i \land j}$ and $\subst{i}{i \lor j}$
correspond to special kinds of degeneracies called
\emph{connections}~\cite{Brown}. The connections $p \subst{i}{i \land
  j}$ and $p \subst{i}{i \lor j}$ can be drawn as the squares:
%
% % Old version using tikzcd:
% \begin{equation*}
%   \begin{tikzcd}[row sep=1.5cm,column sep=1.5cm]
%     a \arrow{r}{p} & b  \\
%     a \arrow[phantom]{ur}{p\subst{i}{i \land j}}
%     \arrow[swap]{r}{p (i0)}
%     \arrow{u}{p (i0)} &
%     a \arrow[swap]{u}{p \subst{i}{j}}
%   \end{tikzcd}
%   \qquad %
%   \begin{tikzcd}[row sep=1.5cm,column sep=1.5cm]
%     b \arrow{r}{p (i1)} & b  \\
%     a \arrow[phantom]{ur}{p\subst{i}{i \lor j}}
%     \arrow[swap]{r}{p}
%     \arrow{u}{p\subst{i}{j}} &
%     b \arrow[swap]{u}{p (i1)}
%   \end{tikzcd}
%   \qquad %
%   \vcenter{\hbox{
%     \begin{tikzpicture}[->, scale=0.75]
%       \draw (0,0) -- node [left] {$j$} (0,1);
%       \draw (0,0) -- node [below] {$i$} (1,0);
%     \end{tikzpicture}}}
% \end{equation*}
%
% % New version using only tikz:
\begin{equation*}
  \begin{tikzpicture}[baseline=30,xscale=2,yscale=2]
    \node (A01) at (0,1) {$a$};
    \node (A11) at (1,1) {$b$};
    \node (A00) at (0,0) {$a$};
    \node (A10) at (1,0) {$a$};
    \node (c) at (0.5,0.5) {$p\subst{i}{i \land j}$};
    \path[->,font=\scriptsize,>=angle 90]
      (A01) edge node[above]{$p$} (A11)
      (A00) edge node[left]{$p(i0)$} (A01)
      (A10) edge node[right]{$p \subst{i}{j}$} (A11)
      (A00) edge node[below]{$p(i0)$} (A10);
  \end{tikzpicture}
  \qquad %
  \begin{tikzpicture}[baseline=30,xscale=2,yscale=2]
    \node (A01) at (0,1) {$b$};
    \node (A11) at (1,1) {$b$};
    \node (A00) at (0,0) {$a$};
    \node (A10) at (1,0) {$b$};
    \node (c) at (0.5,0.5) {$p\subst{i}{i \lor j}$};
    \path[->,font=\scriptsize,>=angle 90]
      (A01) edge node[above]{$p(i1)$} (A11)
      (A00) edge node[left]{$p \subst{i}{j}$} (A01)
      (A10) edge node[right]{$p(i1)$} (A11)
      (A00) edge node[below]{$p$} (A10);
  \end{tikzpicture}
  \qquad %
  \vcenter{\hbox{
    \begin{tikzpicture}[->, scale=0.75]
      \draw (0,0) -- node [left] {$j$} (0,1);
      \draw (0,0) -- node [below] {$i$} (1,0);
    \end{tikzpicture}}}
\end{equation*}
where, for instance, the right-hand side of the left square is
computed as
\[
p \subst{i}{i \land j} (i1) \= p \subst{i}{1 \land j} \= p
\subst{i}{j}
\]
and the bottom and left-hand sides are degenerate.

\subsection{Examples}\label{subsec:examples}

Representing equalities using path types allows novel definitions of
many standard operations on identity types that are usually proved by
identity elimination. For instance, the fact that the images of two
equal elements are equal can be defined as:
\begin{mathpar}
  \inferrule {\Gamma \der a : A \\ \Gamma \der b : A \\
              \Gamma \der f : A \rightarrow B \\
              \Gamma \der p : \Path~A~a~b}
             {\Gamma \der \pabs{i} f \; (p \; i) :
              \Path~B~(f \; a)~(f \; b)} %
\end{mathpar}
This operation satisfies some judgmental equalities that do not hold
judgmentally when the identity type is defined as an inductive family
(see Section 7.2 of~\cite{BCH} for details).

We can also define new operations, for instance, function
extensionality for path types can be proved as:
\begin{mathpar}
  \inferrule {\Gamma \der f : (x : A) \rightarrow B \\
              \Gamma \der g : (x : A) \rightarrow B \\
              \Gamma \der p : (x : A) \rightarrow \Path~B~(f \; x)~(g \; x)}
             {\Gamma \der \pabs{i} \lambda x : A. \; p \; x \; i :
              \Path~((x : A) \rightarrow B)~f~g} %
\end{mathpar}
To see that this is correct we check that the term has the correct
faces, for instance:
\[
(\pabs{i} \lambda x : A. \; p \; x \; i) \; 0 \=
\lambda x : A. \; p \; x \; 0 \=
\lambda x : A. \; f \; x \=
f
\]

We can also justify the fact that singletons are contractible, that
is, that any element in $(x : A) \times (\Path~A~a~x)$ is equal to
$(a,\refl{a})$:
\begin{mathpar}
  \inferrule {\Gamma \der p : \Path~A~a~b}
             {\Gamma \der \pabs{i}(p \; i, \pabs{j} p \; (i \land j)) :
              \Path~((x : A) \times (\Path~A~a~x))~(a,\refl{a})~(b,p)} %
\end{mathpar}

As in the previous work~\cite{BCH,simonlic} we need to add {\em
  composition operations}, defined by induction on the type, in order
to justify the elimination principle for paths.

\section{Systems, composition, and transport}\label{sec:compositions}
% - Partial elements + face lattice
% - Definition + picture of composition
% - Kan Filling
% - Transport, J without defeq.
% - Recursive definition of composition (+ proof of type preservation
%   and that it computes the correct face... do we need something more?)
% - Examples!
% - What about uniformity?
% - Motivation why we need de Morgan algebra and not Boolean algebra
% - Motivation why we need systems/partial elements and not just open boxes?

In this section we define the operation of context \emph{restriction}
which will allow us to describe new geometrical shapes corresponding
to ``sub-polyhedra'' of a cube. Using this we can define the
composition operation. From this operation we will also be able to
define the transport operation and the elimination principle for
$\Path$ types.

\subsection{The face lattice}

The {\em face lattice}, $\FF$, is the distributive lattice generated
by symbols $(i = 0)$ and $(i = 1)$ with the relation $(i = 0) \wedge
(i = 1) = 0_\FF$. The elements of the face lattice, called \emph{face
  formulas}, can be described by the grammar
\[
\begin{array}{lcl}
  \varphi,\psi & ::= & 0_\FF \| 1_\FF \| (i = 0) \| (i = 1)
                    \| \varphi \land \psi \| \varphi \lor \psi \\
\end{array}
\]
There is a canonical lattice map $\II \rightarrow \FF$ sending $i$ to
$(i=1)$ and $1-i$ to $(i=0)$. We write $(r=1)$ for the image of $r :
\II$ in $\FF$ and we write $(r=0)$ for $(1-r = 1)$. We have $(r=1)
\wedge (r=0) = 0_\FF$ and we define the lattice map $\FF \rightarrow
\FF,~\psi\longmapsto \psi \subst{i}{r}$ sending $(i=1)$ to $(r=1)$ and
$(i=0)$ to $(r=0)$.

\medskip

Any element of $\FF$ is the join of the irreducible elements below
it. An irreducible element of this lattice is a \emph{face}, \ie{}, a
conjunction of elements of the form $(i = 0)$ and $(j = 1)$. This
provides a disjunctive normal form for face formulas, and it follows
from this that the equality on $\FF$ is decidable.

\medskip

Geometrically, the face formulas describe ``sub-polyhedra'' of a
cube. For instance, the element $(i=0) \vee (j=1)$ can be seen as the
union of two faces of the square in directions $j$ and $i$. If $I$ is
a finite set of names, we define the {\em boundary} of $I$ as the
element $\partial_I$ of $\FF$ which is the disjunction of all $(i=0)
\vee (i=1)$ for $i$ in $I$. It is the greatest element depending at
most on elements in $I$ which is $< 1_\FF$.

We write $\Gamma \der \psi : \FF$ to mean that $\psi$ is a face
formula using only the names declared in $\Gamma$. We introduce then
the new {\em restriction} operation on contexts:
\[
\begin{array}{lclr}
  \Gamma,\Delta & ::= & \dots \| \Gamma, \varphi & \\
\end{array}
\]
together with the rule:
\begin{mathpar}
  \inferrule {\Gamma \der \varphi : \FF} {\Gamma, \varphi \der}
\end{mathpar}

This allows us to describe new geometrical shapes: as we have seen
above, a type in a context $\Gamma = i : \II, j : \II$ can be thought
of as a square, and a type in the restricted context $\Gamma, \varphi$
will then represent a compatible union of faces of this square. This
can be illustrated by:
%
% % Old version using tikzcd:
% \begin{mathpar}
% \begin{array}{|c|c|}
%   \hline
%   i : \II, (i=0)\vee (i=1) \der A & \begin{tikzcd} A(i0) \; \bullet  & \bullet \; A(i1) \end{tikzcd} \\ \hline
%   i : \II, j : \II, (i=0)\vee (j=1) \der A &
%        \begin{tikzcd}[row sep=1.25cm,column sep=1.25cm]
%             A(i0)(j1) \arrow{r}{A(j1)} & A(i1)(j1) \\
%             A(i0)(j0) \arrow{u}{A(i0)} &
%        \end{tikzcd} \\ \hline
%   i : \II, j : \II, (i=0)\vee (i=1)\vee (j=0) \der A &
%        \begin{tikzcd}[row sep=1.25cm,column sep=1.25cm]
%             A(i0)(j1)   & A(i1)(j1)  \\
%             A(i0)(j0) \arrow[phantom]{ur}{} \arrow[swap]{r}{A(j0)} \arrow{u}{A(i0)} & A(i1)(j0) \arrow[swap]{u}{A(i1)}
%        \end{tikzcd} \\ \hline
% \end{array}
% \end{mathpar}
%
\begin{center}
\setlength{\tabcolsep}{5mm} % separator between columns
\begin{tabular}{|c|c|}
  \hline
  $i : \II, (i=0)\vee (i=1) \der A$ &
  \begin{tikzpicture}[baseline=-2,scale=2]
    \node (A0) at (0,0) {$A(i0) \; \bullet$};%
    \node (A1) at (1,0) {$A(i1) \; \bullet$};%
  \end{tikzpicture}
  \\ \hline
  $i : \II, j : \II, (i=0) \vee (j=1) \der A$ &
  \begin{tikzpicture}[baseline=30,xscale=3,yscale=1.8]
    \node (A01) at (0,1) {$A(i0)(j1)$};
    \node (A11) at (1,1) {$A(i1)(j1)$};
    \node (A00) at (0,0) {$A(i0)(j0)$};
    \path[->,font=\scriptsize,>=angle 90]
      (A01) edge node[above]{$A(j1)$} (A11)
      (A00) edge node[left]{$A(i0)$} (A01);
  \end{tikzpicture}
  \\ \hline
  $i : \II, j : \II, (i=0 )\vee (i=1) \vee (j=0) \der A$ &
  \begin{tikzpicture}[baseline=30,xscale=3,yscale=1.8]
    \node (A01) at (0,1) {$A(i0)(j1)$};
    \node (A11) at (1,1) {$A(i1)(j1)$};
    \node (A00) at (0,0) {$A(i0)(j0)$};
    \node (A10) at (1,0) {$A(i1)(j0)$};
    \path[->,font=\scriptsize,>=angle 90]
      (A00) edge node[left]{$A(i0)$} (A01)
      (A10) edge node[right]{$A(i1)$} (A11)
      (A00) edge node[below]{$A(j0)$} (A10);
  \end{tikzpicture}
  \\ \hline
\end{tabular}
\end{center}
There is a canonical map from the lattice $\FF$ to the congruence
lattice of $\II$, which is distributive~\cite{Balbes}, sending $(i=1)$ to the congruence
identifying $i$ with $1$ (and $1-i$ with $0$) and sending $(i=0)$ to
the congruence identifying $i$ with $0$ (and $1-i$ with $1$). In this
way, any element $\psi$ of $\FF$ defines a congruence
$r = s \mod{\psi}$ on $\II$.

\medskip

This congruence can be described as a substitution if $\psi$ is
irreducible; for instance, if $\psi$ is $(i=0) \wedge (j=1)$ then
$r = s \mod{\psi}$ is equivalent to $r(i0)(j1) = s(i0)(j1)$. The
congruence associated to $\psi = \varphi_0 \vee \varphi_1$ is the meet
of the congruences associated to $\varphi_0$ and $\varphi_1$
respectively, so that we have, \eg{}, $i = 1-j \mod{\psi}$ if
$\varphi_0 = (i=0) \wedge (j=1)$ and $\varphi_1 = (i=1) \wedge (j=0)$.

\medskip

To any context $\Gamma$ we can associate recursively a congruence on
$\II$, the congruence on $\Gamma,\psi$ being the join of the
congruence defined by $\Gamma$ and the congruence defined by $\psi$.
The congruence defined by $()$ is equality in $\II$, and an extension
$x:A$ or $i:\II$ does not change the congruence.  The judgment $\Gamma
\der r = s : \II$ then means that $r = s \mod{\Gamma}$, $\Gamma \der r
: \II$, and $\Gamma \der s : \II$.

In the case where $\Gamma$ does not use the restriction operation, this judgment
means $r = s$ in $\II$. If $i$ is declared in $\Gamma$, then
$\Gamma,(i=0) \der r = s : \II$ is equivalent to
$\Gamma \der r(i0) = s(i0) : \II$. Similarly any context $\Gamma$
defines a congruence on $\FF$ with
$\Gamma,\psi \der \varphi_0 = \varphi_1 : \FF$ being equivalent to
$\Gamma \der \psi \wedge \varphi_0 = \psi \wedge \varphi_1 : \FF$.

\medskip

As explained above, the elements of $\II$ can be seen as formal
representations of elements in the interval $[0,1]$. The elements of
$\FF$ can then be seen as formulas on elements of $[0,1]$.  We have a
simple form of {\em quantifier elimination} on $\FF$: given a
name~$i$, we define $\forall i \colon \FF \to \FF$ as the lattice
morphism sending $(i=0)$ and $(i=1)$ to $0_\FF$, and being the
identity on all the other generators.  If $\psi$ is independent of
$i$, we have $\psi \leqslant \varphi$ if, and only if, $\psi \leqslant
\forall i. \varphi$.  For example, if $\varphi$ is $(i=0) \vee
((i=1)\wedge (j=0))\vee (j=1)$, then $\forall i.\varphi$ is $(j=1)$.
This operation will play a crucial role in
\sect{sec:composition-glueing} for the definition of composition of
glueing.

Since $\FF$ is not a Boolean algebra, we don't have in general
$\varphi = (\varphi \wedge (i=0)) \vee (\varphi \wedge (i=1))$, but we
always have the following decomposition:

\begin{lemma}\label{decomp}
  For any element $\varphi$ of $\FF$ and any name $i$ we have
  \[
  \varphi = (\forall i. \varphi) \vee (\varphi \wedge (i=0)) \vee (\varphi \wedge  (i=1))
  % \varphi = \delta \vee (\varphi \wedge (i=0)) \vee (\varphi \wedge  (i=1))
  \]
  We also have $\varphi\wedge (i=0)\leqslant \varphi(i0)$
and  $\varphi\wedge (i=1)\leqslant \varphi(i1)$.
\end{lemma}

\subsection{Syntax and inference rules for systems}

Systems allow to introduce ``sub-polyhedra'' as compatible unions of
cubes.  The extension to the syntax of dependent type theory with path
types is:
\[
\begin{array}{lcll}
  % \Gamma,\Delta & ::= & \dots \| \Gamma, \varphi & \\ \\
  t,u,A,B       & ::= & \dots & \\
                & \|  &  [ \; \varphi_1 \hmapsto t_1, \dots, \varphi_n \hmapsto t_n \; ]
                      & \hspace{2cm} \text{Systems}
\end{array}
\]
We allow $n = 0$ and get the empty system $[ \; ]$. As explained
above, a context now corresponds in general to the union of sub-faces
of a cube. In~\fig{systems:rules} we provide operations for combining
compatible systems of types and elements, the side condition for these
rules is that $\Gamma \der \varphi_1 \vee \dots \vee \varphi_n = 1_\FF
: \FF$. This condition requires $\Gamma$ to be sufficiently
restricted: for example $\Delta, (i = 0) \vee (i = 1) \der (i = 0)
\vee (i = 1) = 1_\FF$.  The first rule introduces systems of types,
each defined on one $\varphi_l$ and requiring the types to agree
whenever they overlap; the second rule is the analogous rule for
terms.  The last two rules make sure that systems have the correct
faces.  The third inference rule says that that any judgment which is
valid locally at each $\varphi_l$ is valid; note that in particular
$n=0$ is allowed (then the side condition becomes $\Gamma \der 0_\FF =
1_\FF : \FF$).
\begin{figure}[h]
\begin{mdframed}
\begin{mathpar}
  % \inferrule {\Gamma, \varphi_1 \der \mathsf{J} \\ \cdots \\ \Gamma, \varphi_n \der \mathsf{J}}
  %            {\Gamma, \varphi \der \mathsf{J}} %
  % \and %
  \inferrule {\Gamma, \varphi_1 \der A_1 \quad \cdots \quad \Gamma, \varphi_n \der A_n \\
              \Gamma, \varphi_i \land \varphi_j \der A_i = A_j \quad
              (1 \leqslant i, j \leqslant n)}
             {\Gamma \der
             [ \; \varphi_1 \hmapsto A_1, \dots, \varphi_n \hmapsto A_n \; ]} %
  \and %
  \inferrule {\Gamma \der A \\ \Gamma, \varphi_1 \der t_1 : A \quad \cdots \quad
              \Gamma, \varphi_n \der t_n : A \\
              \Gamma, \varphi_i \land \varphi_j \der t_i = t_j : A ~~
              (1 \leqslant i, j \leqslant n)}
             {\Gamma \der
              [ \; \varphi_1 \hmapsto t_1, \dots, \varphi_n \hmapsto t_n \; ] : A} %
  \and %
  \inferrule {\Gamma, \varphi_1 \der \mathsf{J} \quad \cdots \quad
              \Gamma, \varphi_n \der \mathsf{J}}
             {\Gamma \der \mathsf{J}} %
  \and %
%  \inferrule {\Gamma, \varphi_1 \der A = B \hspace{0.5cm} \cdots \hspace{0.5cm}
%              \Gamma, \varphi_n \der A = B}
%             {\Gamma, \varphi_1 \lor \cdots \lor \varphi_n \der A = B} %
%  \and %
%  \inferrule {\Gamma, \varphi_1 \der t = u : A \hspace{0.5cm} \cdots \hspace{0.5cm}
%              \Gamma, \varphi_n \der t = u : A}
%             {\Gamma, \varphi_1 \lor \cdots \lor \varphi_n \der t = u : A} %
%  \and %
  \inferrule {\Gamma \der [ \; \varphi_1 \hmapsto A_1, \dots,
                 \varphi_n \hmapsto A_n \; ]\\
              \Gamma \der \varphi_i = 1_\FF : \FF}
             {\Gamma \der
             [ \; \varphi_1 \hmapsto A_1, \dots, \varphi_n \hmapsto A_n \; ] = A_i} %
  \and %
  \inferrule {\Gamma\der [ \; \varphi_1 \hmapsto t_1, \dots, \varphi_n \hmapsto t_n \; ] : A \\
              \Gamma \der \varphi_i = 1_\FF : \FF}
             {\Gamma \der
             [ \; \varphi_1 \hmapsto t_1, \dots, \varphi_n \hmapsto t_n \; ] = t_i : A} %
\end{mathpar}
\end{mdframed}
\caption{Inference rules for systems with side condition $\Gamma \der
  \varphi_1 \vee \dots \vee \varphi_n = 1_\FF :
  \FF$}\label{systems:rules}
\end{figure}

Note that when $n = 0$ the second of the above rules should be read
as: if $\Gamma \der 0_\FF = 1_\FF : \FF$ and $\Gamma \der A$, then
$\Gamma \der [ \; ] : A$.

We extend the definition of the substitution judgment by $\Delta \der
\sigma : \Gamma,\varphi$ if $\Delta \der \sigma : \Gamma$, $\Gamma
\der \varphi : \FF$, and $\Delta \der \varphi \sigma = 1_\FF : \FF$.

\medskip

If $\Gamma, \varphi \der u : A$, then $\Gamma \der a : A[\varphi
\mapsto u]$ is an abbreviation for $\Gamma \der a:A$ and $\Gamma,
\varphi \der a = u : A$.  In this case, we see this element $a$ as a
witness that the partial element $u$, defined on the ``extent''
$\varphi$ (using the terminology from \cite{Fourman-Scott}), is {\em
  extensible}.  More generally, we write $\Gamma \der a : A[\varphi_1
\mapsto u_1, \dots, \varphi_k \mapsto u_k]$ for $\Gamma \der a : A$
and $\Gamma,\varphi_l \der a = u_l : A$ for $l =
1,\dots,k$.

For instance, if $\Gamma, i : \II \der A$ and $\Gamma, i : \II,
\varphi \der u : A$ where $\varphi = (i=0) \vee (i=1)$ then the
element $u$ is determined by two elements $\Gamma \der a_0 : A(i0)$
and $\Gamma \der a_1 : A(i1)$ and an element $\Gamma,i:\II \der a :
A[(i=0) \mapsto a_0,(i=1) \mapsto a_1]$ gives a path connecting $a_0$
and $a_1$.

\begin{lemma}\label{admissible} The following rules are admissible:\footnote{The
    inference rules with double line are each a pair of rules, because
    they can be used in both directions.}
\begin{mathpar}
  \inferrule {\Gamma \der \varphi \leqslant \psi : \FF \\
              \Gamma, \psi \der \mathsf{J}}
             {\Gamma, \varphi \der \mathsf{J}} %
  \and %
  \mprset{fraction={===}}
  \inferrule {\Gamma, 1_\FF \der \mathsf{J}} {\Gamma \der \mathsf{J}} %
  \and %
  % \inferrule {\Gamma \der \mathsf{J}} {\Gamma, 1_\FF \der \mathsf{J}} %
  % \and %
  \mprset{fraction={===}}
  \inferrule {\Gamma, \varphi, \psi \der \mathsf{J}}
             {\Gamma, \varphi \land \psi \der \mathsf{J}} %
%%   \and %
%%   \inferrule {\Gamma, \varphi \land \psi \der \mathsf{J}}
%%              {\Gamma, \varphi, \psi \der \mathsf{J}} %
\end{mathpar}
Furthermore, if $\varphi$ is independent of $i$, the following rules
are admissible
\begin{mathpar}
  \mprset{fraction={===}}
  \inferrule {\Gamma,i:\II,\varphi \der  \mathsf{J}}
             {\Gamma, \varphi,i:\II \der \mathsf{J}} %
\end{mathpar}
and it follows that we have in general:
\begin{mathpar}
  \inferrule {\Gamma,i:\II,\varphi \der  \mathsf{J}}
             {\Gamma, \forall i.\varphi,i:\II \der \mathsf{J}} %
\end{mathpar}
\end{lemma}

\subsection{Composition operation}

The syntax of compositions is given by:
\[
\begin{array}{lcll}
  t,u,A,B       & ::= & \dots & \\
                & \|  & \comp^i~A~[\varphi\mapsto u]~a_0
                      & \hspace{2cm} \text{Compositions} \\
\end{array}
\]
where $u$ is a system on the extent $\varphi$.

The composition operation expresses that being extensible is preserved
along paths: if a partial path is extensible at $0$, then it is
extensible at $1$.
\begin{mathpar}
  \inferrule {\Gamma \der \varphi : \FF \\ \Gamma, i : \II \der A \\
              \Gamma, \varphi, i : \II \der u : A \\
              \Gamma \der a_0 : A(i0)[ \varphi \mapsto u(i0) ]}
             {\Gamma \der \comp^i~A~[ \varphi \mapsto u ]~a_0 :
              A(i1)[ \varphi \mapsto u(i1) ]}
\end{mathpar}
Note that $\comp^i$ binds $i$ in $A$ and $u$ and that we have in
particular the following equality judgments for systems:
\[
\Gamma \der \comp^i~A~[1_\FF \mapsto u]~a_0 = u(i1) : A(i1)
\]

If we have a substitution $\Delta \der \sigma : \Gamma$, then
\[
(\comp^i~A~[ \varphi \mapsto u ]~a_0)\sigma = %
\comp^j~A(\sigma,i/j)~[\varphi \sigma\mapsto u(\sigma,i/j)]~a_0 \sigma
\]
where $j$ is fresh for $\Delta$, which corresponds semantically to the
{\em uniformity} \cite{BCH,simonlic} of the composition operation.

We use the abbreviation
$[\varphi_1 \mapsto u_1, \dots, \varphi_n \mapsto u_n]$ for
$\left[\bigvee_l\varphi_l \mapsto [\varphi_1 \hmapsto u_1, \dots, \varphi_n
\hmapsto u_n]\right]$
and in particular we write $[]$ for $[ 0_\FF \mapsto [ \; ] ]$.

\begin{example}
  With composition we can justify transitivity of path types:
\begin{mathpar}
  \inferrule {\Gamma \der p : \Path~A~a~b \\ \Gamma \der q : \Path~A~b~c}
             {\Gamma \der \pabs{i} \comp^j~A~[ (i = 0) \mapsto a, (i =
               1) \mapsto q \; j ]~(p \; i) : \Path~A~a~c}
\end{mathpar}
This composition can be visualized as the dashed arrow in the square:
\begin{mathpar}
  % % Old version:
  % \begin{tikzcd}[row sep=1.5cm,column sep=1.5cm]
  %   a \arrow[dashed]{r}{} & c \\
  %   a \arrow[swap]{r}{p \; i} \arrow{u}{a} & b \arrow[swap]{u}{q \; j}
  % \end{tikzcd}
  % \and %
  \begin{tikzpicture}[baseline=30,xscale=2,yscale=2]
    \node (A01) at (0,1) {$a$};
    \node (A11) at (1,1) {$c$};
    \node (A00) at (0,0) {$a$};
    \node (A10) at (1,0) {$b$};
    \path[->,font=\scriptsize,>=angle 90]
      (A00) edge node[left]{$a$} (A01)
      (A10) edge node[right]{$q \; j$} (A11)
      (A00) edge node[below]{$p \; i$} (A10);
    \path[->,font=\scriptsize,>=angle 90,dashed]
      (A01) edge node[above]{$$} (A11);
  \end{tikzpicture}
  \and %
  \vcenter{\hbox{\begin{tikzpicture}[->, scale=0.75]
    \draw (0,0) -- node [left] {$j$} (0,1);
    \draw (0,0) -- node [below] {$i$} (1,0);
  \end{tikzpicture}}}
\end{mathpar}
\end{example}

\subsection{Kan filling operation}
\label{sec:kan-filling}

As we have connections we also get Kan filling operations from
compositions:
\[
\Gamma, i : \II \der \Comp^i~A~[\varphi\mapsto u]~a_0 =
             \comp^j~A\subst{i}{i\wedge j}~[\varphi \mapsto
             u\subst{i}{i \wedge j}, (i=0) \mapsto a_0]~a_0 : A
\]
where $j$ is fresh for $\Gamma$. The element
$\Gamma, i : \II \der v = \Comp^i~A~[ \varphi \mapsto u ]~a_0 : A$
satisfies:
\begin{align*}
  \Gamma &\der v(i0) = a_0 : A(i0) %
  & %
  \Gamma &\der v(i1) = \comp^i~A~[\varphi \mapsto u]~a_0 : A(i1) %
  & %
  \Gamma, \varphi, i : \II &\der v = u : A %
\end{align*}
This means that we can not only compute the lid of an open box but
also its filling. If $\varphi$ is the boundary formula on the names
declared in $\Gamma$, we recover the Kan operation for cubical
sets~\cite{Kan55}.

\subsection{Equality judgments for composition}\label{sec:composition}

The equality judgments for $\comp^i~C~[ \varphi \mapsto u ]~a_0$ are
defined by cases on the type $C$ which depends on $i$, i.e., $\Gamma,
i : \II \der C$. The right hand side of the definitions are all equal
to $u(i1)$ on the extent $\varphi$ by the typing rule for
compositions. There are four cases to consider:

\subsubsection*{Product types, $C = (x : A) \rightarrow B$}

Given $\Gamma, \varphi,i : \II \der \mu : C$ and $\Gamma \der
\lambda_0 : C(i0)[ \varphi \mapsto \mu(i0) ]$ the composition will be
of type $C(i1)$.  For $\Gamma \der u_1 : A(i1)$, we first let:
\begin{align*}
  w &= \Comp^i~A\subst{i}{1-i}~[]~u_1 %
  && (\text{in context }\Gamma, i : \II\text{ and of type
  }A\subst{i}{1-i}) \\
  v &= w\subst{i}{1-i} %
  && (\text{in context }\Gamma, i : \II\text{ and of type }A)
\end{align*}
Using this we define the equality judgment:
\[
  \Gamma \der (\comp^i~C~[ \varphi \mapsto \mu ]~\lambda_0)~u_1 = %
  \comp^i~B \subst{x}{v}~[ \varphi \mapsto \mu \; v]~(\lambda_0 \;
  v(i0)) : B\subst x v (i1)
\]

\subsubsection*{Sum types, $C = (x : A) \times B$}

Given $\Gamma, \varphi, i : \II \der w : C$ and $\Gamma \der w_0 :
C(i0)[ \varphi \mapsto w(i0) ]$ we let:
\begin{align*}
  a &= \Comp^i~A~[ \varphi \mapsto w.1]~w_0.1 %
  && (\text{in context }\Gamma, i : \II\text{ and of type }A)\\
  c_1 &= \comp^i~A~[ \varphi \mapsto w.1]~w_0.1 %
  && (\text{in context }\Gamma\text{ and of type }A(i1))\\
  c_2 &= \comp^i~B\subst{x}{a}~[\varphi \mapsto w.2]~w_0.2%
  && (\text{in context }\Gamma\text{ and of type }B\subst{x}{a}(i1))
\end{align*}
From which we define:
\[
\Gamma \der \comp^i~C~[ \varphi \mapsto w]~w_0 = (c_1,c_2) : C(i1)
\]

\subsubsection*{Natural numbers, $C = \NN$}

In this we define $\comp^i~C~[\varphi\mapsto n]~n_0$ by recursion:
\begin{align*}
  &\Gamma \der \comp^i~C~[ \varphi \mapsto \0 ]~\0 = \0 : C
  \\
  &\Gamma \der \comp^i~C~[ \varphi \mapsto \suc{n} ]~(\suc{n_0}) =
  \suc{(\comp^i~C~[ \varphi \mapsto n ]~n_0)} : C
\end{align*}

\subsubsection*{Path types, $C = \Path~A~u~v$}

Given $\Gamma, \varphi , i : \II\der p : C$ and $\Gamma \der p_0 :
C(i0)[ \varphi \mapsto p(i0) ]$ we define:
\[
\Gamma \der \comp^i~C~[ \varphi \mapsto p]~p_0 = \pabs{j} \comp^i~A~[
\varphi \mapsto p \; j, (j=0) \mapsto u, (j=1) \mapsto v ]~(p_0 \; j)
: C(i1)
\]

%% \subsubsection*{Systems}

%% In the case of a system $\Gamma, i : \II \der C = [ \; \varphi_1
%% \hmapsto A_1, \dots, \varphi_n \hmapsto A_n \; ]$ with $\Gamma, i :
%% \II \der \varphi_1 \vee \dots \vee \varphi_n = 1_\FF : \FF$ and given
%% $\Gamma, \varphi, i : \II \der u : C$ and $\Gamma \der u_0 : C(i0)[
%% \varphi \mapsto u(i0) ]$, we set:
%% \begin{multline*}
%%   \Gamma \der \comp^i~C~[ \varphi \mapsto u]~u_0 =\\
%%   [ \; \varphi_1 \hmapsto (\comp^i~A_1~[ \varphi \mapsto u]~u_0),
%%   \dots, \varphi_n \hmapsto (\comp^i~A_n~[ \varphi \mapsto u]~u_0) \;
%%   ] : C (i1) [ \varphi \mapsto u(i1) ]
%% \end{multline*}
% Note that this equation is provable since it holds trivially when
% restricting the context to each $\varphi_i$.

\subsection{Transport}

Composition for $\varphi = 0_\FF$ corresponds to transport:
\[
\Gamma \der \transport^i~A~a = \comp^i~A~[]~a : A(i1)
\]

Together with the fact that singletons are contractible,
from~\sect{subsec:examples}, we get the elimination principle for
$\Path$ types in the same manner as explained for identity
types in Section 7.2 of~\cite{BCH}.

\section{Derived notions and operations}\label{sec:contr}

This section defines various notions and operations that will be used
for defining compositions for the $\glue$ operation in the next
section. This operation will then be used to define the composition
operation for the universe and to prove the univalence axiom.

\subsection{Contractible types}
\label{sec:contractible-types}

We define
$\isContr~A = (x : A) \times ((y : A) \rightarrow \Path~A~x~y)$.
A proof of $\isContr~A$ witnesses the fact that $A$ is {\em contractible}.

Given $\Gamma \der p : \isContr~A$ and $\Gamma, \varphi \der u : A$ we
define the operation\footnote{This expresses that the restriction map
  $\Gamma,\varphi\rightarrow \Gamma$ has the left lifting property
  w.r.t.\ any ``trivial fibration'', \ie{}, contractible extensions
  $\Gamma,x:A\rightarrow\Gamma$.  The restriction maps
  $\Gamma,\varphi\rightarrow\Gamma$ thus represent ``cofibrations''
  while the maps $\Gamma,x:A\rightarrow\Gamma$ represent
  ``fibrations''.}
\[
\Gamma \der \ext~p~[\varphi \mapsto u] = \comp^i~A~[\varphi \mapsto
p.2~u~i]~p.1 : A[\varphi \mapsto u]
\]

Conversely, we can state the following characterization of
contractible types:

\begin{lemma}\label{contrinv}
  Let $\Gamma \der A$ and assume that we have one operation
\begin{mathpar}
\inferrule {\Gamma, \varphi \vdash u : A}
           {\Gamma \vdash \ext~[\varphi \mapsto u] : A[\varphi\mapsto u]} %
%  \and %
%  \inferrule{\Gamma \vdash u : A}
%            {\Gamma \vdash \ext_2~u : \Path~A~u~(\ext_1~[1_\FF \mapsto u])} %
\end{mathpar}
then we can find an element in $\isContr~A$.
\end{lemma}

\begin{proof}
  We define $x = \ext~[] : A$ and prove that any element $y : A$ is
  path equal to $x$. For this, we introduce a fresh name $i : \II$ and
  define $\varphi = (i=0)\vee (i=1)$ and
  $u = [(i=0) \mapsto x, (i=1) \mapsto y]$. Using this we obtain
  $\Gamma, i : \II \der v = \ext~[\varphi \mapsto u] : A[\varphi \mapsto u]$.
  In this way, we get a path $\pabs{i} {\ext~[\varphi \mapsto u]}$
  connecting $x$ and $y$.
\end{proof}

% Old proof:
% \begin{proof}
%   We define $a = \ext_1~[] : A$ and prove that any element $x : A$ is
%   path equal to $a$. For this, we introduce a fresh name $i : \II$ and
%   define $\varphi = (i=0)\vee (i=1)$ and
%   $u = [(i=0) \mapsto a, (i=1) \mapsto x]$. Using this we obtain
%   $\Gamma, i : \II \der v = \ext_1~[\varphi \mapsto u] : A$. Further,
%   from the operation $\ext_2$ we get a path between $a$ and $v(i0)$,
%   and a path between $v(i1)$ and $x$. By composing these we get a path
%   connecting $a$ and $x$.
%  \end{proof}

\subsection{The~$\pres{}$ operation}\label{sec:two-deriv-oper}

The $\pres{}$ operation states that the image of a composition is path
equal to the composition of the respective images, so that any
function {\em preserves} composition, up to path equality.

\begin{lemma}\label{pres}
We have an operation:
\begin{mathpar}
  \inferrule {\Gamma, i : \II \der f : T \rightarrow A \\
              \Gamma \der \varphi : \FF\\
              \Gamma, \varphi, i : \II \der t : T \\
              \Gamma \der t_0 : T(i0)[\varphi \mapsto t(i0)]}
             {\Gamma \der \pres{i}~f~[\varphi \mapsto t]~t_0 :
             (\Path~A(i1)~c_1~c_2)[\varphi \mapsto \pabs{j} (f \; t)(i1)]}

\end{mathpar}
where $c_1 = \comp^i~A~[\varphi\mapsto f \; t]~(f(i0) \; t_0)$ and $c_2 = f(i1) \; (\comp^i~T~[\varphi \mapsto t]~t_0)$.
%
% \begin{mathpar}
%   c_1 = \comp^i~A~[\varphi\mapsto f \; t]~(f(i0) \; t_0) %
%   \and %
%   c_2 = f(i1) \; (\comp^i~T~[\varphi \mapsto t]~t_0) %
% \end{mathpar}
\end{lemma}
\begin{proof}
  Let $\Gamma\der a_0 = f(i0) \; t_0 : A(i0)$ and
%  $\Gamma,\varphi,i:\II \der a = f \; t : A$, together with
  $\Gamma,i:\II \der v = \Comp^i~T~[\varphi \mapsto t]~t_0 : T$.
  We take
$\pres{i}~f~[\varphi \mapsto t]~t_0 = \pabs{j} \comp^i~A~[\varphi\vee (j=1) \mapsto f \;v]~a_0.$
\end{proof}

Note that $\pres{i}$ binds $i$ in $f$ and $t$.

% Old version of the lemma:
% \begin{lemma}\label{pres}
%   If we have
%   $\Gamma, i : \II \der f : T \rightarrow A,~\Gamma \der \varphi$ and
%   $\Gamma, \varphi, i : \II \der t : T$ with
%   $\Gamma \der t_0 : T(i0)[\varphi \mapsto t(i0)]$ then we can build
% %
% \[
% \Gamma \der \pres(f,[\varphi \mapsto t],t_0) :
% \Path~A(i1)~\left(\comp^i~A~[\varphi\mapsto a]~a_0\right)~\left(f(i1)
%   \; t_1\right)[\varphi \mapsto \pabs{j} a(i1)]
% \]
% %
% where
% \[ \begin{array}{rlclll}
% \der & a_0 & = & f(i0) \; t_0 & : & A(i0) \\
% i : \II, \varphi \der & a & = & f \; t & : & A \\
% \der & t_1 & = & \comp^i~T~[\varphi \mapsto t]~t_0 & : & T
% \end{array} \]
% % Furthermore, we have
% % $$\Gamma,\varphi\der \pres(f,[\varphi\mapsto t],t_0) = \lam{j} a(i1)$$
% \end{lemma}
% \begin{proof}
% %%   We assume given
% %%   $\Gamma, i : \II \der f : T \rightarrow A,~~\Gamma\der \varphi$
% %%   and $\Gamma,\varphi,i:\II\der t:T$ with
% %%   $\Gamma\der t_0:T(i0)[\varphi\mapsto t(i0)]$.
%   We define $\Gamma\der a_0 = f(i0) \; t_0 : A(i0)$ and
%   $\Gamma,i:\II,\varphi \der a = f \; t : A$, and
%   \begin{mathpar}
%     \Gamma,i:\II \der u = \Comp^i~A~[\varphi \mapsto a]~a_0 : A %
%     \and %
%     \Gamma,i:\II \der v = \Comp^i~T~[\varphi \mapsto t]~t_0 : T %
%   \end{mathpar}
%   We define then
%   $$
%   \Gamma \der \pres(f,[\varphi \mapsto t],t_0) = \pabs{j}
%   \comp^i~A~[\varphi \mapsto f \; t,(j=0) \mapsto u,(j=1) \mapsto
%   f \; v]~a_0
%   $$
% \end{proof}

\subsection{The~$\eq$ operation}
\label{sec:eq-operation}

We define
$\isEquiv~T~A~f = (y : A) \rightarrow \isContr~((x : T) \times
\Path~A~y~(f \; x))$
and $\Equiv~T~A = (f : T \rightarrow A) \times \isEquiv~T~A~f$. If
$\myeq{} : \Equiv~T~A$ and $t : T$, we may write $\myeq{} \; t$ for
$\myeq{}.1 \; t$.

\begin{lemma}\label{equiv}
  If $\Gamma \der \myeq{} : \Equiv~T~A$, we have an operation
  \begin{mathpar}
  \inferrule {\Gamma,\varphi \der t : T \\
              \Gamma \der a : A~~~~~\Gamma,\varphi \der p : \Path~A~a~(\myeq{}~t)}
             {\Gamma \der \eq~\myeq{}~[\varphi \mapsto (t,p)]~a : ((x : T)
              \times \Path~A~a~(\myeq{} \; x))[\varphi \mapsto (t,p)]} %
  \end{mathpar}
  Conversely, if $\Gamma \der \myeq{} : T \rightarrow A$ and we have
  such an operation, then we can build a proof that $\myeq{}$ is an
  equivalence.
\end{lemma}

\begin{proof}
  We define
  $\eq~\myeq{}~[\varphi \mapsto (t,p)]~a = \ext~(\myeq{}.2 \;a)~[\varphi
  \mapsto (t,p)]$
  using the $\ext$ operation defined above. The second statement
  follows from Lemma \ref{contrinv}.
\end{proof}

\section{Glueing}\label{sec:glue}
% - Glueing
% - compGlue

In this section, we introduce the glueing operation. This operation
expresses that to be ``extensible'' is invariant by equivalence. From
this operation, we can define a composition operation for universes,
and prove the univalence axiom.

\subsection{Syntax and inference rules for glueing}\label{sec:glueing}

We introduce the \emph{glueing} construction at type and term level
by:
\[
\begin{array}{lcll}
  t,u,A,B       & ::= & \dots & \\
                & \|  & \Glue~[\varphi\mapsto (T,\myeq{})]~A & \hspace{2cm} \text{Glue type} \\
                & \|  & \glue~[\varphi\mapsto t]~u & \hspace{2cm} \text{Glue term} \\
                & \|  & \ugl~[\varphi\mapsto \myeq{}]~u & \hspace{2cm} \text{Unglue term} \\
\end{array}
\]
We may write simply $\ugl~b$ for $\ugl~[\varphi\mapsto \myeq{}]~b$.
The inference rules for these are presented in~\fig{glueing:rules}.

\begin{figure}[h]
\begin{mdframed}
\begin{mathpar}
  \inferrule {\Gamma \der A \\ \Gamma,\varphi \der T \\
              \Gamma,\varphi \der \myeq{} : \Equiv~T~A}
             {\Gamma \der \Glue~[\varphi \mapsto (T,\myeq{})]~A} %
  \and %
  \inferrule {\Gamma \der b : \Glue~[\varphi \mapsto (T,\myeq{})]~A}
             {\Gamma \der \ugl \; b : A[\varphi \mapsto \myeq{} \;
               b]} %
  \and %
  \inferrule {\Gamma,\varphi \der \myeq{} : \Equiv~T~A \\
              \Gamma,\varphi \der t : T \\
              \Gamma \der a : A[\varphi \mapsto \myeq{} \; t]}
             {\Gamma \der \glue~[\varphi \mapsto t]~a :
              \Glue~[\varphi \mapsto (T,\myeq{})]~A} %

  \and %
  \inferrule {\Gamma \der T \\ \Gamma \der \myeq{} : \Equiv~T~A}
             {\Gamma \der \Glue~[1_\FF \mapsto (T,\myeq{})]~A = T} %
  \and %
  \inferrule {\Gamma \der t : T\\
              \Gamma \der \myeq{} : \Equiv~T~A}
             {\Gamma \der \glue~[1_\FF \mapsto t]~(f \; t) = t : T} %
  \and %
  \inferrule {\Gamma \der b : \Glue~[\varphi \mapsto (T,\myeq{})]~A}
             {\Gamma \der b = \glue~[\varphi \mapsto b]~(\ugl \; b) :
              \Glue~[\varphi \mapsto (T,\myeq{})]~A} %
  \and %
  \inferrule {\Gamma,\varphi \der \myeq{} : \Equiv~T~A \\
              \Gamma, \varphi \der t : T \\
              \Gamma \der a : A[\varphi \mapsto \myeq{} \; t ]}
             {\Gamma \der \ugl \; (\glue~[\varphi\mapsto t]~a) = a : A}
   %           A[\varphi \mapsto \myeq{} \; t ]} %
\end{mathpar}
\end{mdframed}
\caption{Inference rules for glueing}\label{glueing:rules}
\end{figure}

It follows from these rules that if
$\Gamma \der b:\Glue~[\varphi \mapsto (T,\myeq{})]~A$, then
$\Gamma, \varphi \der b : T$.

In the case $\varphi = (i=0) \vee (i=1)$ the glueing operation can be
illustrated as the dashed line in:
\begin{mathpar}
% % Old version:
% \begin{tikzcd}[row sep=1.5cm,column sep=1.5cm]
%   T_0 \arrow[dashed]{r}{}
%   % The reflectbox part rotates the sim symbols 90 degrees
%   \arrow[swap]{d}{\myeq{}(i0)}[swap]{\reflectbox{\rotatebox[origin=c]{90}{$\sim$}}} &
%   T_1 \arrow{d}{\myeq{}(i1)}[swap]{\reflectbox{\rotatebox[origin=c]{90}{$\sim$}}} \\
%   A(i0) \arrow[swap]{r}{A} & % [squiggly]  & % [squiggly]
%   A(i1)
% \end{tikzcd}
%   \and %
  \begin{tikzpicture}[xscale=2.5,yscale=2]
    \node (A01) at (0,1) {$T_0$};
    \node (A11) at (1,1) {$T_1$};
    \node (A00) at (0,0) {$A(i0)$};
    \node (A10) at (1,0) {$A(i1)$};
    \path[->,font=\scriptsize,>=angle 90]
      (A01) edge node[left]{$f(i0)$} node[right]{\reflectbox{\rotatebox[origin=c]{90}{$\sim$}}} (A00)
      (A11) edge node[left]{\reflectbox{\rotatebox[origin=c]{90}{$\sim$}}} node[right]{$f(i1)$} (A10)
      (A00) edge node[below]{$A$} (A10);
    \path[->,font=\scriptsize,>=angle 90,dashed]
      (A01) edge node[above]{$$} (A11);
  \end{tikzpicture}
\end{mathpar}
This illustrates why the operation is called glue: it \emph{glues}
together along a partial equivalence the partial type $T$ and the
total type $A$ to a total type that extends $T$.

\begin{remark*}
In general $\Glue~[\varphi\mapsto (T,\myeq{})]~A$ can be illustrated as:
\begin{mathpar}
% \includegraphics{gluefig}
  % \begin{tikzcd}[column sep=0.5cm,row sep = 0.5cm]
  %   T\arrow{drr}{\myeq{}}[swap]{\reflectbox{\rotatebox[origin=c]{45}{$\sim$}}}
  %   \arrow[bend right=20,two heads]{dddddrr}
  %   \arrow[dashed]{rrrr} & & & &
  %   \Glue~[\varphi\mapsto (T,\myeq{})]~A \arrow[dashed]{drr}{\ugl}
  %         \arrow[dashed,bend right=20,two heads]{dddddrr}  & & \\
  %     & & A\arrow[two heads]{dddd}\arrow{rrrr} & & & & A\arrow[two heads]{dddd} \\
  %     &  & &  \pb\ \\
  %     &  & &  \\
  %     &  & &  \\
  %   & & \Gamma, \varphi \arrow[tail]{rrrr}{} & & & & \Gamma &
  % \end{tikzcd}
  \begin{tikzpicture}
    [x={(1cm,0cm)},y={(0cm,1cm)},z={(1cm,-0.5cm)}, scale=1.5]
    % base
    \node (Gphi) at (0,0,0) {$\Gamma,\varphi$};
    \node (G) at (3,0,0) {$\Gamma$};
    \draw[>->] (Gphi) -- (G);
    % left
    \node (T) at (0,2,-1) {$T$} ;
    \node (A') at (0,2,1) {$A$} ;
    \draw[->>] (A') -- (Gphi) ;
    \draw[->>] (T) -- (Gphi) ;
    \draw[->] (T) -- node [below,sloped] {\scriptsize $\sim$} node [above,sloped] {\scriptsize $f$}(A');
    % right
    \node (A) at (3,2,1) {$A$} ;
    \node (Glue) at (3,2,-1) {$\Glue~[\varphi\mapsto (T,\myeq{})]~A$} ;
    \draw[->>] (A) -- (G) ;
    \draw[->>,dashed] (Glue) -- (G) ;
    \draw[->,dashed] (Glue) -- node [above,sloped] {\scriptsize $\ugl$} (A) ;
    \draw[->,dashed] (T) -- (Glue) ;
    \draw[->] (A') -- (A) ;
    % pb corners
    \draw (0.2,1.6,0.8) -- (0.4,1.6,0.8) -- (0.4,1.8,0.9);
    %\draw [dashed] (0.2,1.6,-0.8) -- (0.4,1.6,-0.8) -- (0.4,1.8,-0.9);
    \draw [dashed] (0.2,1.7,-0.85) -- (0.4,1.7,-0.85) -- (0.4,1.9,-0.95);
  \end{tikzpicture}
\end{mathpar}
This diagram suggests that a construction similar to $\Glue$ also
appears in the simplicial set model. Indeed, the proof of
Theorem~3.4.1 in~\cite{KL} contains a similar diagram where
$\overline{E}_1$ corresponds to $\Glue~[\varphi\mapsto
(T,\myeq{})]~A$.
\end{remark*}

\begin{example}
  \label{exa:equivtopath}
  Using glueing we can construct a path from an equivalence
  $\Gamma\der \myeq{} : \Equiv~A~B$ by defining
  \[
  \Gamma,i:\II \der E =
  \Glue~[(i=0) \mapsto (A,\myeq{}),(i=1) \mapsto (B,\ident_B)]~B
  \]
  so that $E(i0) = A$ and $E(i1) = B$, where $\ident_B : \Equiv~B~B$ is
  defined as:
  \[
  \ident_B = (\lambda x : B. \; x,
  \lambda x : B. \; ((x,\refl{x}),
  \lambda u : (y : B) \times \Path~B~x~y. \;
  \pabs{i} (u.2~i,\pabs{j} u.2 \; (i\wedge j))))
  \]
  % We have then $\Gamma,i:\II,(i=0)\der E = A$ and
  % $\Gamma,i:\II,(i=1)\der E = B$, so that $E(i0) = A$ and $E(i1) = B$.
  %
  % \medskip
  In \sect{sec:universe} we introduce a universe of types $\U$ 
  and  we will be able to define a function of type
  $(A~B : \U) \rightarrow \Equiv~A~B \rightarrow \Path~\U~A~B$ by:
  \[
  \lambda A~B:\U. \; \lambda \myeq{} : \Equiv~A~B. \;
  \pabs{i} \Glue~[(i=0) \mapsto (A,\myeq{}),(i=1) \mapsto (B,\ident_B)]~B
  \]
\end{example}

\subsection{Composition for glueing}\label{sec:composition-glueing}

We assume
$\Gamma, i : \II \der B = \Glue~[\varphi \mapsto (T,\myeq{})]~A$, and
define the composition in~$B$. In order to do so, assume
\begin{align*}
  % \Gamma,i:\II,\psi\der
  % b:B % \Gamma,i:\II,\psi\der b = & \; \glue~[\varphi\mapsto
  % b]~a):B
  \Gamma,\psi,i : \II \der b : B %
  && %
  \Gamma\der b_0 :B(i0)[\psi\mapsto b(i0)]
  % \Gamma\der b_0 = & \; \glue~[\varphi(i0)\mapsto
  % b_0]~a_0):B(i0)[\psi\mapsto b(i0)]
\end{align*}
and define:
%
% \begin{align*}
% %  \Gamma,i:\II,\psi,\varphi\der t : T\label{eq:glue-t}\\
% %  \Gamma,i:\II,\psi\der a = & \; \ugl~ b: A[\varphi\mapsto \myeq{} \;  b]%\label{eq:glue-a}
%    \Gamma,\psi,i:\II \der a = & \; \ugl~ b: A[\varphi\mapsto \myeq{} \;  b]%\label{eq:glue-a}
%   \\
% %  \Gamma,\varphi(i0)\der t_0 : T(i0)[\psi\mapsto t(i0)]\label{eq:glue-t0}\\
%   \Gamma\der a_0 = & \; \ugl~ b_0: A(i0)[\varphi(i0)\mapsto \myeq{}(i0) \; b_0,\psi\mapsto a(i0)]%\label{eq:glue-a0}
% \end{align*}
%
\begin{align*}
  a &= \ugl~ b %
  && (\text{in context }\Gamma, \psi, i : \II\text{ and of type }A[\varphi\mapsto \myeq{} \;  b])\\
  a_0 &= \ugl~ b_0 %
  && (\text{in context }\Gamma\text{ and of type }A(i0)[\varphi(i0)\mapsto \myeq{}(i0) \; b_0,\psi\mapsto a(i0)])
\end{align*}
The following provides the algorithm for composition
$\comp^i~B~[\psi\mapsto b]~b_0 = b_1$ of type $B(i1)[\psi \mapsto
b(i1)]$.
%A detailed description of this algorithm can be found in the
%Appendix~\ref{appendix:composition-glueing}.
%
%$\delta = \forall i. \varphi$.
%
\[ \begin{array}{lcll}
\delta & = & \forall i.\varphi & \hspace{1cm} \Gamma \\
a_1' & = & \comp^i~A~[\psi\mapsto a]~a_0 & \hspace{1cm} \Gamma \\
t_1' & = & \comp^i~T~[\psi\mapsto b]~b_0 & \hspace{1cm} \Gamma,\delta\\
\omega & = &\pres{i}~\myeq{}~[\psi\mapsto b]~b_0 & \hspace{1cm} \Gamma,\delta \\
%a_1'' & = & \comp^j~A(i1)~[\delta\mapsto \omega~j,\psi\mapsto a(i1)]~a_1' & \hspace{1cm} \Gamma \\
(t_1,\alpha) & = & \eq~\myeq{}(i1)~[\delta \mapsto (t'_1,\omega),
                            \psi \mapsto (b(i1),\lam{j}{a_1'})]~a_1' & \hspace{1cm} \Gamma,\varphi(i1)\\
%\alpha & = & (\eq~\myeq{}(i1)~[\delta\mapsto (t'_1,\omega),\psi\mapsto (b(i1),\lam{j}{a_1'})]~a_1').2 & \hspace{1cm} \Gamma,\varphi(i1) \\
%% (t_1,\alpha) & = & \eq~\myeq{}(i1)~[\delta \mapsto t'_1,\psi \mapsto b(i1)]~a_1'' & \hspace{1cm} \Gamma,\varphi(i1)\\
a_1 & = & \comp^j~A(i1)~[\varphi(i1)\mapsto \alpha~j,\psi\mapsto a(i1)]~a_1' & \hspace{1cm} \Gamma \\
b_1 & = & \glue~[\varphi(i1)\mapsto t_1]~a_1 & \hspace{1cm} \Gamma \\
\end{array} \]

We can check that whenever $\Gamma, i: \II \der \varphi = 1_\FF : \FF$
the definition of $b_1$ coincides with $\comp^i~T~[\psi\mapsto
b]~b_0$, which is consistent with the fact that $B = T$ in this case.

\medskip

In the next section we will use the $\glue$ operation to define the
composition for the universe and to prove the univalence axiom.

\section{Universe and the univalence axiom}\label{sec:universe}
% - Any line of gives an equivalence
% - comp for U using glue
% - Univalence using glue
As in \cite{MLTT72}, we now introduce a universe $\U$ \`a la Russell
by reflecting all typing rules and
\begin{mathpar}
  \inferrule { \Gamma \der {} } {\Gamma \der \U} %
  \and %
  \inferrule {\Gamma \der A : \U} {\Gamma \der A} %
\end{mathpar}
In particular, we have $\Gamma \der \Glue~[\varphi \mapsto
(T,\myeq{})]~A : \U$ whenever $\Gamma \der A : \U$, $\Gamma,\varphi \der
T : \U$, and $\Gamma, \varphi \der \myeq{} : \Equiv~T~A$.

\subsection{Composition for the universe}\label{subsec:compU}

In order to describe the composition operation for the universe we
first have to explain how to construct an equivalence from a line in
the universe. Given $\Gamma \der A,~\Gamma \der B$, and
$\Gamma,i:\II \der E$, such that $E(i0) = A$ and $E(i1) = B$, we will
construct $\eq^i~E : \Equiv~A~B$. In order to do this we first define
% \begin{align*}
%   \Gamma \der f &= \lambda x:A.~\transport^i~E~x : A \to B
%   \\
%   \Gamma \der g &= \lambda y:B.~(\transport^i~E(i/1-i)~y)(i/1-i) : B
%   \to A
%   \\
%   \Gamma,i:\II \der u &= \lambda x:A.~\Comp^i~E~[]~x : A \to E
%   \\
%   \Gamma,i:\II \der v &= \lambda y:B.~(\Comp^i~E(i/1-i)~[]~y)(i/1-i) :
%   B \to E
% \end{align*}
\begin{align*}
  f &= \lambda x:A.~\transport^i~E~x %
  && (\text{in context }\Gamma\text{ and of type } A \to B)
  \\
  g &= \lambda y:B.~(\transport^i~E(i/1-i)~y)(i/1-i) %
  && (\text{in context }\Gamma\text{ and of type } B \to A)
  \\
  u &= \lambda x:A.~\Comp^i~E~[]~x %
  && (\text{in context }\Gamma, i:\II\text{ and of type } A \to E)
  \\
  v &= \lambda y:B.~(\Comp^i~E(i/1-i)~[]~y)(i/1-i) %
  && (\text{in context }\Gamma, i:\II\text{ and of type } B \to E)
\end{align*}
such that:
\begin{align*}
  u(i0) &= \lambda x:A.x
  &
  u(i1) &= f
  &
  v(i0) &= g
  &
  v(i1) &= \lambda y:B.y
\end{align*}

%% such that $u(i1) = f$ and $u(i0) = \lambda x:A.x$ and $v(i0) = g$ and
%% $v(i1) = \lambda y:B.y$.

%% The definitions are
%% \[ \begin{array}{lcl}
%% f & = & \lambda x:A.\transport^i~E~x \\
%% g & = & \lambda y:B.\transport^i~E(i/1-i)~y \\
%% u & = & \lambda x:A.\Comp^i~E~[]~x \\
%% v & = & \lambda y:B.\Comp^i~E(i/1-i)~[]~y
%% \end{array}
%% \]

We will now prove that $f$ is an equivalence. Given $y : B$ we see
that $(x : A) \times \Path~B~y~(f \; x)$ is inhabited as it contains
the element $(g \; y, \pabs j \theta_0 (i1))$ where
% \[
% \Gamma, i : \II, j : \II \der \theta_0 =%
% \Comp^i~E~[(j=0) \mapsto v \; y,(j=1) \mapsto u \; (g \; y)]~(g \; y)
% : E.
% \]
\[
\theta_0 =%
\Comp^i~E~[(j=0) \mapsto v \; y,(j=1) \mapsto u \; (g \; y)]~(g \; y).
\]
Next, given an element $(x,\beta)$ of $(x : A) \times \Path~B~y~(f \;
x)$ we will construct a path from $(g \; y, \pabs j \theta_0 (i1))$ to
$(x,\beta)$.  Let
\[
\theta_1 =%
\left(\Comp^i~{E \subst i {1-i}}~[(j=0) \mapsto (v \; y) \subst i {1-i}, %
(j=1) \mapsto (u \; x) \subst i {1-i}]~(\beta \; j) \right)\subst i {1-i}
\]
and $\omega = \theta_1 (i0)$ so $\Gamma, i : \II, j : \II \der
\theta_1 : E$, $\omega (j0) = g \; y$, and $\omega (j1) = x$.  And
further with
\[
\delta = %
\comp^i~E~%
[ (k=0) \mapsto \theta_0%
, (k=1) \mapsto \theta_1%
, (j=0) \mapsto v \; y%
, (j=1) \mapsto u \; \omega \subst j k%
]~%
\omega \subst j {j \land k}
\]
we obtain
\[
\pabs k (\omega \subst j k, \pabs j \delta) : %
\Path~((x:A) \times \Path~B~y~(f \; x))~(g \; y, \pabs j
\theta_0 (i1))~(x,\beta)
\]
as desired. This concludes the proof that $f$ is an equivalence and
thus also the construction of $\eq^i~E : \Equiv~A~B$.

\medskip

Using this we can now define the composition for the universe:
\[
\Gamma \der \comp^i~\U~[\varphi\mapsto E]~A_0 =
\Glue~[\varphi\mapsto(E(i1),\eq^i~E(i/1-i))]~A_0 : \U
\]

\begin{remark*}
  Given $\Gamma, i : \II \der E$ we can also get an equivalence in
  $\Equiv~A~B$ (where $A = E (i0)$ and $B = E(i1)$) with a less direct
  description by
  \[
  \Gamma \der \transport^i~(\Equiv~A~E)~\ident_A : \Equiv~A~B
  \]
  where $\ident_A$ is the identity equivalence as given in
  Example~\ref{exa:equivtopath}.
\end{remark*}

\subsection{The univalence axiom}
\label{sec:univalence-axiom}

Given $B = \Glue~[\varphi \mapsto (T,\myeq{})]~A$ the map $\ugl : B
\rightarrow A$ extends $\myeq{}$, in the sense that $\Gamma, \varphi
\der \ugl~b = \myeq{}~b : A$ if $\Gamma \der b : B$.

\begin{theorem}\label{thm:unglueequiv}
  The map $\ugl : B \rightarrow A$ is an equivalence.
\end{theorem}

\begin{proof}
By Lemma \ref{equiv} it suffices to construct
\begin{align*}
  \tilde{b} : B[\psi \mapsto b] %
  &&
  \tilde{\alpha} : \Path~A~u~(\ugl~\tilde{b})[\psi \mapsto \alpha] %
\end{align*}
% \begin{mathpar}
%   \tilde{b} : B[\psi \mapsto b] %
%   \and %
%   \alpha : \Path~A~u~(\ugl~\tilde{b})[\psi \mapsto \pabs{i} u] %
% \end{mathpar}
%
given $\Gamma, \psi \der b : B$ and $\Gamma\der u:A$ and
$\Gamma, \psi \der \alpha : \Path~A~u~(\ugl~b)$.

Since $\Gamma, \varphi \der \myeq{} : T \rightarrow A$ is an
equivalence and
\begin{align*}
  \Gamma,\varphi,\psi \der b : T %
  &&
  \Gamma,\varphi,\psi \der \alpha:\Path~A~u~(f~b) %
\end{align*}
% \begin{mathpar}
%   \Gamma,\varphi,\psi \der b : T %
%   \and %
%   \Gamma,\varphi,\psi \der \myeq{}~b = u : A %
% \end{mathpar}
we get, using Lemma \ref{equiv}
\begin{align*}
  \Gamma,\varphi \der t : T[\psi \mapsto b] %
  &&
  \Gamma,\varphi \der \beta : \Path~A~u~(\myeq{}~t)~[\psi \mapsto \alpha] %
\end{align*}
% \begin{mathpar}
%   \Gamma,\varphi \der t : T[\psi \mapsto b] %
%   \and %
%   \Gamma,\varphi \der \beta : \Path~A~u~(\myeq{}~t)~[\psi \mapsto \pabs{i} u] %
% \end{mathpar}
We then define
$\tilde{a} = \comp^i~A~[\varphi \mapsto \beta~i,\psi \mapsto \alpha~i]~u$,
and using this we conclude by letting
$\tilde{b} = \glue~[\varphi \mapsto t]~\tilde{a}$ and
$\tilde{\alpha} = \Comp^i~A~[\varphi \mapsto \beta~i, \psi \mapsto \alpha~i]~u$.
\end{proof}

%The following corollary is a possible way to state the univalence
%axiom
%
\begin{corollary}
  \label{cor:equiv-contr}
  For any type $A : \U$ the type $C = (X : \U) \times \Equiv~X~A~$ is
  contractible.\footnote{This formulation of the univalence axiom
    can be found in the message of Martín Escardó in:\\
    \url{https://groups.google.com/forum/\#!msg/homotopytypetheory/HfCB_b-PNEU/Ibb48LvUMeUJ}\\
    This is also used in the (classical) proofs of the univalence
    axiom, see Theorem~3.4.1 of~\cite{KL} and Proposition~2.18 of
    \cite{Cisinski}, where an operation similar to the glueing
    operation appears implicitly.}
\end{corollary}
\begin{proof}
  It is enough by Lemma~\ref{contrinv} to show that any partial
  element $\varphi \der (T,\myeq{}):C$ is path equal to the restriction
  of a total element. The map $\ugl$ extends $\myeq{}$ and is an
  equivalence by the previous theorem. Since any two elements of the
  type $\isEquiv~X~A~\myeq{}.1$ are path equal, this shows that any
  partial element of type $C$ is path equal to the restriction of a
  total element.  We can then conclude by
  Theorem~\ref{thm:unglueequiv}.
\end{proof}

%% Using this we get the following statement (see Theorem 5.8.4
%% in~\cite{hott-book}).

\begin{corollary}[Univalence axiom]\label{univalence} For any term
  \[
  t : (A~B : \U) \rightarrow \Path~\U~A~B \rightarrow \Equiv~A~B
  \]
  the map $t~A~B:\Path~\U~A~B\rightarrow \Equiv~A~B$ is an
  equivalence.
% \footnote{We also have a direct formal proof
% that the map
% for obtaining an equivalence from a path in the
% universe, in~\sect{subsec:compU}, is an equivalence. This proof has been formally verified
% inside the system using the \Haskell{} implementation. For
%   details see:
%   \url{https://github.com/mortberg/cubicaltt/blob/master/examples/univalence.ctt}}
%   Any map $\Path~U~A~B \rightarrow \Equiv(A,B)$ is an equivalence if
%   it is defined for all $A$ and $B$ in $U$.
\end{corollary}
\begin{proof}
Both $(X : \U) \times \Path~\U~A~X$ and $(X : \U) \times \Equiv~A~X$
are contractible. Hence the result follows from Theorem 4.7.7 in
\cite{hott-book}.
\end{proof}

Two alternative proofs of univalence can be found in
Appendix~\ref{sec:univ-from-glue}.

\section{Semantics}
\label{sec:semantics}

In this section we will explain the semantics of the type theory under
consideration in cubical sets. We will first review how cubical sets,
as a presheaf category, yield a model of basic type theory, and then
explain the additional so-called composition structure we have to
require to interpret the full cubical type theory.

\subsection{The category of cubes and cubical sets}
\label{sec:cubecat}

Consider the monad $\dM$ on the category of sets associating to each
set the free de~Morgan algebra on that set.  The \emph{category of
  cubes} $\CC$ is the small category whose objects are finite subsets
$I,J,K,\dots$ of a fixed, discrete, and countably infinite set, called
\emph{names}, and a morphism $\Hom(J,I)$ is a map $I \to \dM(J)$.
Identities and compositions are inherited from the Kleisli category of
$\dM$, \ie{}, the identity on $I$ is given by the unit $I \to \dM
(I)$, and composition $f g \in \Hom (K,I)$ of $g \in \Hom(K,J)$ and $f
\in \Hom(J,I)$ is given by $ \mu_K \circ \dM (g) \circ f$ where $\mu_K
\colon \dM (\dM (K)) \to \dM (K)$ denotes multiplication of $\dM$.  We
will use $f,g,h$ for morphisms in $\CC$ and simply write $f \colon J
\to I$ for $f \in \Hom(J,I)$.  We will often write unions with commas
and omit curly braces around finite sets of names, \eg{}, writing
$I,i,j$ for $I \cup \set {i,j}$ and $I-i$ for $I-\set{i}$ etc.

If $i$ is in $I$ and $b$ is $0_\II$ or $1_\II$, we have maps $(ib)$ in
$\Hom(I-i,I)$ whose underlying map sends $j \neq i$ to itself and $i$
to $b$.  A \emph{face map} is a composition of such maps. A
\emph{strict map} $\Hom(J,I)$ is a map $I \to \dM(J)$ which never
takes the value $0_\II$ or $1_\II$. Any $f$ can be uniquely written as
a composition $f = gh$ where $g$ is a face map and $h$ is strict.

\begin{definition}
  \label{def:cset}
  A \emph{cubical set} is a presheaf on $\CC$.
\end{definition}

Thus, a cubical set $\Gamma$ is given by sets $\Gamma (I)$ for each $I
\in \CC$ and maps (called restrictions) $\Gamma (f) \colon \Gamma (I)
\to \Gamma (J)$ for each $f \colon J \to I$.  If we write $\Gamma (f)
(\rho) = \rho f$ for $\rho \in \Gamma (I)$ (leaving the $\Gamma$
implicit), these maps should satisfy $\rho \, \id_I = \rho$ and $(\rho
f) g = \rho (f g)$ for $f \colon J \to I$ and $g \colon K \to J$.

Let us discuss some important examples of cubical sets.  Using the
canonical de~Morgan algebra structure of the unit interval, $[0,1]$,
we can define a functor
\begin{equation}
  \label{eq:georeal}
  \CC \to \Top, \quad I \mapsto [0,1]^I.
\end{equation}
If $u$ is in $[0,1]^I$ we can think of $u$ as an environment giving
values in $[0,1]$ to each $i \in I$, so that $iu$ is in $[0,1]$ if $i
\in I$.  Since $[0,1]$ is a de~Morgan algebra, this extends uniquely
to $r u$ for $r \in \dM (I)$. So any $f \colon J \to I$ in $\CC$
induces $f \colon [0,1]^J \to [0,1]^I$ by $i(fu) = (if)u$.

To any topological space $X$ we can associate its \emph{singular
  cubical set} $\sing(X)$ by taking $\sing(X) (I)$ to be the set of
continuous functions $[0,1]^I \to X$.

For a finite set of names $I$ we get the formal cube $\yoneda I$ where
$\yoneda \colon \CC \to [\CC^\op,\Set]$ denotes the Yoneda embedding.
Note that since $\Top$ is cocomplete the functor in~\eqref{eq:georeal}
extends to a cocontinuous functor assigning to each cubical set its
\emph{geometric realization} as a topological space, in such a way that
$\yoneda I$ has $[0,1]^I$ as its geometric realization.

The formal interval $\II$ induces a cubical set given by $\II (I) =
\dM (I)$.  The face lattice $\FF$ induces a cubical set by taking as
$\FF (I)$ to be those $\varphi \in \FF$ which only use symbols in $I$.
The restrictions along $f \colon J \to I$ are in both cases simply
\emph{substituting} the symbols $i \in I$ by $f(i) \in
\dM(J)$.

As any presheaf category, cubical sets have a subobject classifier
$\Omega$ where $\Omega (I)$ is the set of sieves on $I$ (i.e.,
subfunctors of $\yoneda I$).  Consider the natural transformation
$(\cdot=1) \colon \II \to \Omega$ where for $r \in \II (I)$, $(r = 1)$
is the sieve on $I$ of all $f \colon J \to I$ such that $r f = 1_\II$.
The image of $(\cdot = 1)$ is $\FF \to \Omega$, assigning to each
$\varphi$ the sieve of all $f$ with $\varphi f = 1_\FF$.

\subsection{Presheaf semantics}
\label{sec:presheaf-semantics}

The category of cubical sets (with morphisms being natural
transformations) induce---as does any presheaf category---a category
with families (\cwf)~\cite{Dybjer96} where the category of contexts
and substitutions is the category of cubical sets.  We will review the
basic constructions but omit verification of the required equations
(see, \eg{}, \cite{Hofmann97,simonlic,BCH} for more details).

\subsubsection*{Basic presheaf semantics}

As already mentioned the category of (semantic) contexts and
substitutions is given by cubical sets and their maps.  In this
section we will use $\Gamma,\Delta$ to denote cubical sets and
(semantic) substitutions by $\sigma \colon \Delta \to \Gamma$,
overloading previous use of the corresponding meta-variables to
emphasize their intended role.

Given a cubical set $\Gamma$, the types $A$ in context $\Gamma$,
written $A \in \Ty(\Gamma)$, are given by sets $A \rho$ for each $I
\in \CC$ and $\rho \in \Gamma (I)$ together with restriction maps $A
\rho \to A (\rho f)$, $u \mapsto u f$ for $f \colon J \to I$
satisfying $u \, \id_I = u$ and $(u f) g = u (f g) \in A (\rho f g)$
if $g \colon K \to J$. Equivalently, $A \in \Ty (\Gamma)$ are the
presheaves on the category of elements of $\Gamma$. For a type $A \in
\Ty(\Gamma)$ its terms $a \in \Ter (\Gamma; A)$ are given by families
of elements $a \rho \in A \rho$ for each $I \in \CC$ and $\rho \in
\Gamma (I)$ such that $(a\rho) f = a (\rho f)$ for $f \colon J \to I$.
Note that our notation leaves a lot implicit; \eg{}, we should have
written $A (I,\rho)$ for $A \rho$; $A (I,\rho,f)$ for the restriction
map $A \rho \to A (\rho f)$; and $a (I,\rho)$ for $a \rho$.

For $A \in \Ty(\Gamma)$ and $\sigma \colon \Delta \to \Gamma$ we
define $A \sigma \in \Ty(\Delta)$ by $(A \sigma) \rho = A (\sigma
\rho)$ and the induced restrictions.  If we also have $a \in \Ter
(\Gamma;A)$, we define $a \sigma \in \Ter (\Delta; A \sigma)$ by $(a
\sigma) \rho = a (\sigma \rho)$.  For a type $A \in \Ty (\Gamma)$ we
define the cubical set $\Gamma.A$ by $(\Gamma.A) (I)$ being the set of
all $(\rho,u)$ with $\rho \in \Gamma (I)$ and $u \in A \rho$;
restrictions are given by $(\rho,u) f = (\rho f, u f)$.  The first
projection yields a map $\pp \colon \Gamma.A \to \Gamma$ and the
second projection a term $\qq \in \Ter (\Gamma.A; A \pp)$.  Given
$\sigma \colon \Delta \to \Gamma$, $A \in \Ty(\Gamma)$, and $a \in
\Ter(\Delta;A \sigma)$ we define $(\sigma,a) \colon \Delta \to
\Gamma.A$ by $(\sigma,a) \rho = (\sigma \rho, a \rho)$.  For $u \in
\Ter (\Gamma;A)$ we define $[u] = (\id_\Gamma, u) \colon \Gamma \to
\Gamma.A$.

The basic type formers are interpreted as follows.  For $A \in \Ty
(\Gamma)$ and $B \in \Ty (\Gamma.A)$ define $\Sigma_\Gamma (A,B) \in
\Ty (\Gamma)$ by letting $\Sigma_\Gamma (A,B) \rho$ contain all pairs
$(u,v)$ where $u \in A \rho$ and $v \in B (\rho,v)$; restrictions are
defined as $(u,v) f = (uf, vf)$.  Given $w \in \Ter
(\Gamma;\Sigma(A,B))$ we get $w.1 \in \Ter (\Gamma;A)$ and $w.2 \in
\Ter(\Gamma;B[w.1])$ by $(w.1) \rho = \pp (w \rho)$ and $(w.2) \rho =
\qq (w \rho)$ where $\pp (u,v) = u$ and $\qq (u,v) = v$ are the
set-theoretic projections.

Given $A \in \Ty (\Gamma)$ and $B \in \Ty (\Gamma.A)$ the dependent
function space $\Pi_\Gamma (A,B) \in \Ty(\Gamma)$ is defined by
letting $\Pi_\Gamma (A,B) \rho$ for $\rho \in \Gamma (I)$ contain all
families $w = (w_f \mid J \in \CC, f \colon J \to I)$ where
\[
w_f \in \prod_{u \in A (\rho f)} B (\rho f, u) \quad \text{such that}
\quad (w_f \, u) g = w_{fg} (u g) \quad \text{for } u \in A (\rho
f),~g\colon K\to J.
\]
The restriction by $f \colon J \to I$ of such a $w$ is defined by $(w
f)_g = w_{fg}$.  Given $v \in \Ter (\Gamma.A;B)$ we have
$\lambda_{\Gamma; A} v \in \Ter (\Gamma;\Pi (A,B))$ given by
$((\lambda v) \rho)_f \, u = v (\rho f, u)$.  Application $\app(w,u) \in
\Ter (\Gamma;B[u])$ of $w \in \Ter(\Gamma;\Pi (A,B))$ to $u \in
\Ter(\Gamma;A)$ is defined by
\begin{equation}
  \label{eq:app}
  \app(w,u) \rho = (w \rho)_{\id_I} (u \rho) \in (B [u]) \rho.
\end{equation}

Basic data types like the natural numbers can be interpreted as
discrete presheaves, \ie{}, $\sNN \in \Ty (\Gamma)$ is given by $\sNN
\rho = \NAT$; the constants are interpreted by the lifts of the
corresponding set-theoretic operations on $\NAT$.  This concludes the
outline of the basic \cwf\ structure on cubical sets.

\begin{remark*}
  %\label{rem:annotations}
  Following Aczel~\cite{Aczel99} we will make use of that our semantic
  entities are actual sets in the ambient set theory.  This will allow
  us to interpret syntax in Section~\ref{sec:interpretation} with
  fewer type annotations than are usually needed for general
  categorical semantics of type theory
  (see~\cite{Streicher91}). E.g., the definition of application
  $\app (w,u) \rho$ as defined in~\eqref{eq:app} is independent of
  $\Gamma$, $A$ and $B$, since set-theoretic application is a (class)
  operation on all sets.  Likewise, we don't need annotations for
  first and second projections.  But note that we will need the type
  $A$ for $\lambda$-abstraction for $(\lambda_{\Gamma;A} v) \rho$ to
  be a set by the replacement axiom.
\end{remark*}

\subsubsection*{Semantic path types}

Note that we can consider any cubical set $X$ as $X' \in \Ty (\Gamma)$
by setting $X' \rho = X (I)$ for $\rho \in \Gamma (I)$.  We will
usually simply write $X$ for $X'$.  In particular, for a cubical set
$\Gamma$ we can form the cubical set $\Gamma.\II$.

For $A \in \Ty (\Gamma)$ and $u,v \in \Ter (\Gamma;A)$ the semantic
path type $\sPath^\Gamma_A (u,v) \in \Ty (\Gamma)$ is given by: for
$\rho \in \Gamma (I)$, $\sPath_A (u,v) \rho$ consists of equivalence
classes $\pabs i w$ where $i \notin I$, $w \in A (\rho \deg_i)$ such
that $w (i0) = u \rho$ and $w (i1) = v \rho$; two such elements $\pabs
i w$ and $\pabs {j} {w'}$ are equal if{f} $w \subst i j = w'$.  Here
$\deg_i \colon I,i \to I$ is induced by the inclusion $I \subseteq
I,i$ and $\subst i j$ setting $i$ to $j$.  We define $(\pabs i w) f =
\pabs j {w (f,i/j)}$ for $f \colon J \to I$ and $j \notin J$.  For $r
\in \II (I)$ we set $(\pabs i w) \, r = w \subst i r$.  Both
operations, name abstraction and application, lift to terms, \ie{}, if
$w \in \Ter (\Gamma.\II;A)$, then $\spabs w \in \Ter (\Gamma; \sPath_A
(w [0], w[1]))$ given by $(\spabs w) \rho = \pabs i w (\rho \deg_i)$
for a fresh $i$; also if $u \in \Ter (\Gamma;\sPath_A(a,b))$ and $r
\in \Ter (\Gamma; \II)$, then $u ~ r \in \Ter (\Gamma;A)$ defined as
$(u ~r) \rho = (u \rho) ~(r \rho)$.

\subsubsection*{Composition structure}

For $\varphi \in \Ter (\Gamma;\FF)$ we define the cubical set
$\Gamma,\varphi$ by taking $\rho \in (\Gamma,\varphi) (I)$ if{f} $\rho
\in \Gamma (I)$ and $\varphi \rho = 1_\FF \in \FF$; the restrictions
are those induced by $\Gamma$.  In particular, we have $\Gamma,1 =
\Gamma$ and $\Gamma,0$ is the empty cubical set.  (Here, $0 \in \Ter
(\Gamma;\FF)$ is $0 \rho = 0_\FF$ and similarly for $1_\FF$.)  Any
$\sigma \colon \Delta \to \Gamma$ gives rise to a morphism
$\Delta,\varphi \sigma \to \Gamma, \varphi$ which we also will denote
by $\sigma$.

If $A \in \Ty(\Gamma)$ and $\varphi \in \Ter(\Gamma;\FF)$, we define a
\emph{partial element of $A \in \Ty (\Gamma)$ of extent $\varphi$} to
be an element of $\Ter (\Gamma,\varphi; A \iota_\varphi)$ where
$\iota_\varphi \colon \Gamma,\varphi \hookrightarrow \Gamma$ is the
inclusion. So, such a partial element $u$ is given by a family of
elements $u \rho \in A \rho$ for each $\rho \in \Gamma (I)$ such that
$\varphi \rho = 1$, satisfying $(u \rho) f = u (\rho f)$ whenever
$f \colon J \to I$.  Each $u \in \Ter (\Gamma;A)$ gives rise to the
partial element $u \iota \in \Ter(\Gamma,\varphi; A\iota)$; a partial
element is \emph{extensible} if it is induced by such an element of
$\Ter(\Gamma;A)$.

For the next definition note that if $A \in \Ty (\Gamma)$, then $\rho
\in \Gamma(I)$ corresponds to $\rho \colon \yoneda I \to \Gamma$ and
thus $A \rho \in \Ty(\yoneda I)$; also, any $\varphi \in \FF(I)$
corresponds to $\varphi \in \Ter(\yoneda I;\FF)$.

\begin{definition}
  \label{def:comp-struct}
  A \emph{composition structure} for $A \in \Ty (\Gamma)$ is given by
  the following operations.  For each $I$, $i \notin I$, $\rho \in
  \Gamma (I,i)$, $\varphi \in \FF(I)$, $u$ a partial element of $A
  \rho$ of extent $\varphi$, and $a_0 \in A \rho (i0)$ with $a_0 f =
  u_{(i0) f}$ for all $f \colon J \to I$ with $\varphi f = 1_\FF$
  (i.e., $a_0 \iota_\varphi = u (i0) $ if $a_0$ is considered as
  element of $\Ter (\yoneda I; A \rho (i0))$), we require
  \[
  \scomp (I,i,\rho,\varphi,u,a_0) \in A \rho (i1)
  \]
  such that for any $f \colon J \to I$ and $j \notin J$,
  \begin{equation*}
    (\scomp (I,i,\rho,\varphi,u,a_0)) f = %
    \scomp (J,j,\rho (f,i=j),\varphi f,u (f,i=j),a_0 f),
  \end{equation*}
  and $\scomp (I,i,\rho,1_\FF,u,a_0) = u_{(i1)}$.
\end{definition}

A type $A \in \Ty(\Gamma)$ together with a composition structure
$\scomp$ on $A$ is called a \emph{fibrant type}, written
$(A,\scomp) \in \FTy(\Gamma)$.  We will usually simply write
$A \in \FTy(\Gamma)$ and $\scomp_A$ for its composition structure.
But observe that $A \in \Ty(\Gamma)$ can have different composition
structures. Call a cubical set $\Gamma$ \emph{fibrant} if it is a
fibrant type when $\Gamma$ considered as type $\Gamma \in \Ty (\top)$
is fibrant where $\top$ is a terminal cubical set.  A prime example of
a fibrant cubical set is the singular cubical set of a topological
space (see Appendix~\ref{sec-singular}).

\begin{theorem}
  \label{thm:1}
  The \cwf\ on cubical sets supporting dependent products, dependent sums,
  and natural numbers described above can be extended to fibrant types.
\end{theorem}
\begin{proof}
  For example, if $A \in \FTy (\Gamma)$ and $\sigma \colon \Delta \to
  \Gamma$, we set
  \[
  \scomp_{A \sigma} (I,i,\rho,\varphi,u,a_0) = \scomp_A (I,i,\sigma
  \rho,\varphi,u,a_0)
  \]
  as the composition structure on $A \sigma$ in $\FTy(\Delta)$.  Type
  formers are treated analogously to their syntactic counterpart given
  in \sect{sec:compositions}. Note that one also has to check that all
  equations between types are also preserved by their associated
  composition structures.
\end{proof}

Note that we can also, like in the syntax, define a composition
structure on $\sPath_A (u,v)$ given that $A$ has one.

\subsubsection*{Semantic glueing}

Next we will give a semantic counterpart to the $\Glue$ construction.
To define the semantic glueing as an element of $\Ty(\Gamma)$ it is
not necessary that the given types have composition structures or that
the functions are equivalences; this is only needed later to give the
composition structure.  Assume $\varphi \in \Ter(\Gamma;\FF)$, $T \in
\Ty (\Gamma,\varphi)$, $A \in \Ty (\Gamma)$, and $w \in
\Ter(\Gamma,\varphi; T \to A \iota)$ (where $A \to B$ is $\Pi
(A,B\pp)$).
\begin{definition}
  \label{def:semantic-glue}
  The \emph{semantic glueing} $\sGlue_\Gamma (\varphi,T,A,w) \in \Ty
  (\Gamma)$ is defined as follows. For $\rho \in \Gamma (I)$, we let
  $u \in \sGlue (\varphi,T,A,w) \rho$ if{f} either
  \begin{itemize}
  \item $u \in T \rho$ and $\varphi \rho = 1_\FF$; or
  \item $u = \sglue (\varphi \rho, t, a)$ and $\varphi \rho \neq
    1_\FF$, where $t \in \Ter (\yoneda I, \varphi \rho; T \rho)$ and
    $a \in \Ter(\yoneda I; A \rho)$ such that $\app (w \rho,t) = a
    \iota \in \Ter (\yoneda I, \varphi \rho; A \rho \iota)$.
  \end{itemize}
  For $f \colon J \to I$ we define the restriction $u f$ of $u \in
  \sGlue (\varphi,T,A,w)$ to be given by the restriction of $T \rho$
  in the first case; in the second case, i.e., if $\varphi \rho \neq
  1_\FF$, we let $u f = \sglue (\varphi \rho, t, a)f = t_f \in T \rho
  f$ in case $\varphi \rho f = 1_\FF$, and otherwise $u f = \sglue
  (\varphi \rho f, t f, a f)$.
\end{definition}

Here $\sglue$ was defined as a constructor; we extend $\sglue$ to any
$t \in \Ter (\yoneda I; T \rho)$, $a \in \Ter (\yoneda I; A \rho)$
such that $\app (w \rho,t) = a$ (so if $\varphi \rho = 1_\FF$) by
$\sglue (1_\FF, t, a) = t_{\id_I}$.  This way any element of $\sGlue
(\varphi,T,A,w) \rho$ is of the form $\sglue (\varphi \rho, t, a)$ for
suitable $t$ and $a$, and restriction is given by $(\sglue (\varphi
\rho, t, a)) f = \sglue (\varphi \rho f, t f, a f)$.  Note that we get
\begin{equation}
  \label{eq:sglue}
  \sGlue_\Gamma (1_\FF, T, A, w) = T \text{ and } %
  (\sGlue_\Gamma (\varphi, T, A, w)) \sigma = %
  \sGlue_\Delta (\varphi \sigma, T \sigma, A \sigma, w \sigma)
  %\text{ for }\sigma \colon \Delta \to \Gamma.
\end{equation}
for $\sigma \colon \Delta \to \Gamma$.  We define $\sugl(\varphi,w)
\in \Ter (\Gamma.\sGlue (\varphi,T,A,w); A \pp)$ by
\begin{align*}
  \sugl (\varphi,w) (\rho, t) &= \app(w \rho, t)_{\id_I} \in A \rho
  &&\text{whenever }\varphi \rho = 1_\FF, \text{ and}
  \\
  \sugl (\varphi,w) (\rho, \sglue (\varphi,t,a)) &= a &&
  \text{otherwise,}
\end{align*}
where $\rho \in \Gamma (I)$.

\begin{definition}
  \label{def:equiv-structure}
  For $A, B \in \Ty (\Gamma)$ and $w \in \Ter (\Gamma; A \to B)$ an
  \emph{equivalence structure} for $w$ is given by the following
  operations such that for each
  \begin{itemize}
  \item $\rho \in \Gamma (I)$,
  \item $\varphi \in \FF (I)$,
  \item $b \in B \rho$, and
  \item partial elements $a$ of $A\rho$ and $\omega$ of $\sPath_B(\app
    (w \rho,a),b \iota) \rho$ with extent $\varphi$,
  \end{itemize}
  we are given
  \[
  \sequiv_0 (\rho,\varphi,b,a,\omega) \in A \rho, \text{ and a path }
  \sequiv_1 (\rho,\varphi,b,a,\omega) \text{ between }\app (w \rho,
  \sequiv_0 (\rho,\varphi,b,a,\omega)) \text{ and }b
  \]
  such that $\sequiv_0 (\rho,\varphi,b,a,\omega) \iota = a$,
  $\sequiv_1 (\rho,\varphi,b,a,\omega) \iota = \omega$ (where $\iota
  \colon \yoneda I, \varphi \to \yoneda I$) and for any $f \colon J
  \to I$ and $\nu = 0,1$:
  \[
  (\sequiv_\nu (\rho,\varphi,b,a,\omega)) f = \sequiv_\nu (\rho f,\varphi
  f,b f,a f,\omega f).
  \]
\end{definition}

Following the argument in the syntax we can use the equivalence
structure to explain a composition for $\sGlue$.
\begin{theorem}
  \label{thm:gluefibrant}
  If $A \in \FTy (\Gamma)$, $T \in \FTy(\Gamma,\varphi)$, and we have
  an equivalence structure for $w$, then we have a composition
  structure for $\sGlue(\varphi,T,A,w)$ such that the
  equations~\eqref{eq:sglue} also hold for the respective composition
  structures.
\end{theorem}

\subsubsection*{Semantic universes}

Assuming a Grothendieck universe of small sets in our ambient set
theory, we can define $A \in \Ty_0 (\Gamma)$ if{f} all $A \rho$ are
small for $\rho \in \Gamma(I)$; and $A \in \FTy_0 (\Gamma)$ if{f} $A
\in \Ty_0 (\Gamma)$ when forgetting the composition structure of $A$.

\begin{definition}
  \label{def:semantic-universe}
  The semantic universe $\UU$ is the cubical set defined by $\UU (I) =
  \FTy_0 (\yoneda I)$; restriction along $f \colon J \to I$ is simply
  substitution along $\yoneda f$.
\end{definition}

We can consider $\UU$ as an element of $\Ty (\Gamma)$.  As such we
can, as in the syntactic counterpart, define a composition structure
on $\UU$ using semantic glueing, so that $\UU \in \FTy (\Gamma)$.
Note that semantic glueing preserves smallness.

For $T \in \Ter (\Gamma; \UU)$ we can define decoding $\El T \in
\FTy_0 (\Gamma)$ by $(\El T) \rho = (T \rho) \, \id_I$ and likewise for
the composition structure.  For $A \in \FTy_0 (\Gamma)$ we get its
code $\code A \in \Ter (\Gamma;\UU)$ by setting $\code A \rho \in
\FTy_0 (\yoneda I)$ to be given by the sets $(\code A \rho) f = A
(\rho f)$ and likewise for restrictions and composition structure.
These operations satisfy $\El \code A = A$ and $\code {\El T} = T$.

\subsection{Interpretation of the syntax}
\label{sec:interpretation}

Following \cite{Streicher91} we define a partial interpretation
function from raw syntax to the \cwf\ with fibrant types given in the
previous section.

To interpret the universe rules \`a la Russell we assume two
Grothendieck universes in the underlying set theory, say \emph{tiny}
and \emph{small} sets.  So that any tiny set is small, and the set of tiny
sets is small.  For a cubical set $X$ we define $\FTy_0 (X)$ and
$\FTy_1 (X)$ as in the previous section, now referring to tiny and
small sets, respectively.  We get semantic universes $\UU_i (I) =
\FTy_i (\yoneda I)$ for $i=0,1$; we identify those with their lifts to
types.  As noted above, these lifts carry a composition structure, and
thus are fibrant.  We also have $\UU_0 \subseteq \UU_1$ and thus $\Ter
(X;\UU_0) \subseteq \Ter (X;\UU_1)$.  Note that coding and decoding
are, as set-theoretic operations, the same for both universes.  We get
that $\code {\UU_0} \in \Ter (X;\UU_1)$ which will serve as the
interpretation of $\U$.

In what follows, we define a partial interpretation function of raw
syntax: $\den{\Gamma}$, $\den{\Gamma;t}$, and $\den{\Delta;\sigma}$ by
recursion on the raw syntax.  Since we want to interpret a universe
\`a la Russell we cannot assume terms and types to have different
syntactic categories.  The definition is given
%in \fig{fig:interpretation}
below and should be read such that the interpretation is defined
whenever all interpretations on the right-hand sides are defined
\emph{and} make sense; so, e.g., for
$\den{\Gamma}.\El{\den{\Gamma;A}}$ below, we require that
$\den{\Gamma}$ is defined and a cubical set, ${\den{\Gamma;A}}$ is
defined, and $\El{\den{\Gamma;A}} \in \FTy (\den\Gamma)$.  The
interpretation for raw contexts is given by:
\begin{align*}
  %%% Contexts
  \den{\emptyctxt} &= \top %
  & %
  \den{\Gamma, x : A} &= \den{\Gamma}.\!\El{\den{\Gamma; A}} &&
  \text{if } x \notin \dom(\Gamma)
  \\
  \den{\Gamma,\varphi} &= \den{\Gamma},\den{\Gamma;\varphi}
  & %
  \den{\Gamma, i : \II} &= \den{\Gamma}.\II && \text{if } i
  \notin \dom(\Gamma)
\end{align*}
where $\top$ is a terminal cubical set and in the last equation $\II$
is considered as an element of $\Ty (\den\Gamma)$.  When defining
$\den{\Gamma;t}$ we require that $\den{\Gamma}$ is defined and a
cubical set; then $\den{\Gamma;t}$ is a (partial) family of sets
$\den{\Gamma;t} (I,\rho)$ for $I\in \CC$ and $\rho \in \den{\Gamma}
(I)$ (leaving $I$ implicit in the definition).  We define:
\begingroup
\allowdisplaybreaks
\begin{align*}
  %%%
  %%% Types
  % Univ
  \den{\Gamma;\U} &= \code {\UU_0} \in \Ter (\den{\Gamma};\UU_1)
  \\
  % nat
  \den{\Gamma;\NN} &= \code {\sNN} \in \Ter (\den{\Gamma};\UU_0)
  \\
  % pi
  \den{\Gamma; (x:A) \to B} &= %
  \code {\Pi_{\den{\Gamma}} (\El {\den{\Gamma;A}},\El {\den{\Gamma,
        x : A; B}})}
  % && \text{if } x \notin \dom(\Gamma)
  \\
  % sigma
  \den{\Gamma; (x:A) \times B} &= %
  \code {\Sigma_{\den{\Gamma}} (\El {\den{\Gamma;A}},\El
    {\den{\Gamma, x : A; B}})}
  % && \text{if } x \notin \dom(\Gamma)
  \\
  % path
  \den{\Gamma;\Path~A~a~b} &= %
  \code {\sPath^{\den{\Gamma}}_{\El
      {\den{\Gamma;A}}}(\den{\Gamma;a},\den{\Gamma;b})}
  \\
  % Glue
  \den{\Gamma;\Glue~[\varphi \mapsto (T,\myeq{})]~A} &= \code
  {\sGlue_{\den{\Gamma}} (\den{\Gamma;\varphi},
    \El{\den{\Gamma,\varphi;T}},\El{\den{\Gamma;A}},
    \den{\Gamma,\varphi;\myeq{}})}
  \\
  \den{\Gamma; \lambda x : A. t} &= %
  \lambda_{\den{\Gamma};\El {\den{\Gamma;A}}} (\den{\Gamma,x : A; t})
  % && {\text{if }}x \notin \dom (\Gamma)
  \\
  \den{\Gamma; t ~ u} &= \app(\den {\Gamma;t}, \den {\Gamma; u})
  \\
  \den{\Gamma; \pabs i t} &= \spabs_{\den{\Gamma}} \den{\Gamma,
    i\colon \II; t}
  %&& {\text{if }}i \notin \dom (\Gamma)
  \\
  \den{\Gamma; t ~ r} &= \den{\Gamma;t} \den{\Gamma;r}
\end{align*}%
\endgroup
where for path application, juxtaposition on the right-hand side is
semantic path application.  In the case of a bound variable, we assume
that $x$ (respectively $i$) is a \emph{chosen} variable fresh for
$\Gamma$; if this is not possible the expression is undefined.
Moreover, all type formers should be read as those on fibrant types,
i.e., also defining the composition structure.  In the case of
$\sGlue$, it is understood that the function part, i.e., the fourth
argument of $\sGlue$ in Definition~\ref{def:semantic-glue} is $\pp
\circ \den{\Gamma,\varphi;\myeq{}}$ and the required (by
Theorem~\ref{thm:gluefibrant}) equivalence structure is to be
extracted from $\qq \circ \den{\Gamma,\varphi;\myeq{}}$ as in
\sect{sec:eq-operation}.  In virtue of the remark in
\sect{sec:presheaf-semantics} we don't need type annotations to
interpret applications.  Note that coding and decoding tacitly refer
to $\den{\Gamma}$ as well.  For the rest of the raw terms we also
assume we are given $\rho \in \den{\Gamma}(I)$.  Variables are
interpreted by:
\begin{align*}
  \den{\Gamma, x :A; x} \rho &= \qq (\rho) %
  & %
  \den{\Gamma, x :A; y} \rho &= \den {\Gamma; y} (\pp (\rho))%
  & %
  \den{\Gamma, \varphi; y} \rho &= \den{\Gamma;y} \rho %
\end{align*}
These should also be read to include the case when $x$ or $y$ are name
variables; if $x$ is a name variable, we require $A$ to be $\II$.  The
interpretations of $\den{\Gamma;r} \rho$ where $r$ is not a name and
$\den{\Gamma;\varphi} \rho$ follow inductively as elements of $\II$
and $\FF$, respectively.

% For $\star \in \set{{\land},{\lor}}$ and $\sharp \in
% \set{0_\II,1_\II}$ define as elements of $\II$:
% \begin{align*}
%   \den{\Gamma;1-i} \rho &= 1 - \den{\Gamma;i} \rho %
%   & %
%   \den{\Gamma;r \star s} \rho &= \den{\Gamma;r} \rho \star
%   \den{\Gamma;s} \rho
%   & %
%   \den{\Gamma;\sharp} \rho = \sharp
% \end{align*}
% and analogously for $\varphi$, where we also add (for $b \in \set{0,1}$)
% \[
% \den{\Gamma;(r=b)} \rho = (\den{\Gamma;r} \rho = b)\in \FF.
% \]

Constants for dependent sums are interpreted by:
\begin{align*}
  \den{\Gamma; (t,u)} \rho &= (\den {\Gamma;t} \rho, \den {\Gamma;u}
  \rho)
  &
  \den{\Gamma; t.1} \rho &= \pp (\den{\Gamma; t} \rho)
  &
  \den{\Gamma; t.2} \rho &= \qq (\den{\Gamma; t} \rho)
  % \\
  % \den{\Gamma;0} \rho &= 0 \in \NAT & %
  % \den{\Gamma;\suc t} \rho &= \den{\Gamma; t} \rho + 1 & %
  % \den{\Gamma; \natrec \, a \, b} \rho &= \natrec\rho \, (\den{\Gamma;a}
  % \rho, \den{\Gamma;b} \rho)
\end{align*}
Likewise, constants for $\NN$ will be interpreted by their semantic
analogues (omitted).
% where $\natrec\rho$ on the right refers to an appropriate semantic
% recursion operator.
The interpretations for the constants related to glueing are
\begin{align*}
  \den{\Gamma; \glue \, [\varphi \mapsto t] \, u} \rho &= %
  \sglue (\den{\Gamma;\varphi} \rho, \den{\Gamma, \varphi; t}
  \hat\rho, \den{\Gamma;u} \rho)
  \\
  \den{\Gamma; \ugl \, [\varphi \mapsto \myeq] \, u} \rho &= %
  \sugl (\den{\Gamma;\varphi}, \pp \circ \den{\Gamma;\myeq}) (\rho,
  \den{\Gamma; u} \rho)
\end{align*}
where $\den{\Gamma, \varphi; t} \hat\rho$ is the family assigning
$\den{\Gamma, \varphi; t} (\rho f)$ to $J \in \CC$ and $f \colon J \to
I$ (and $\rho f$ refers to the restriction given by $\den{\Gamma}$
which is assumed to be a cubical set). Partial elements are
interpreted by
\[
\den{\Gamma; [ \; \varphi_1 \hmapsto u_1, \dots, \varphi_n \hmapsto
  u_n \; ]} \rho = \den{\Gamma,\varphi_i;u_i} \rho \qquad \text{if }
\den{\Gamma;\varphi_i} \rho = 1_\FF,
\]
where for this to be defined we additionally assume that all
$\den{\Gamma,\varphi_i;u_i}$ are defined and
$\den{\Gamma,\varphi_i;u_i} \rho' = \den{\Gamma,\varphi_j;u_j} \rho'$
for each $\rho' \in \den{\Gamma} (I)$ with $\den{\Gamma;\varphi_i
  \land \varphi_j} \rho' = 1_\FF$.

Finally, the interpretation of composition is given by
\[
\den{\Gamma;\comp^i~A~[ \varphi \mapsto u ]~a_0} \rho = %
\scomp_{\El {\den{\Gamma, i : \II;A}}} (I,j,\rho',\den{\Gamma;\varphi}
\rho, \den{\Gamma,\varphi,i \colon
  \II; u} \rho' , \den{\Gamma;a_0} \rho)
\]
if $i \notin \dom(\Gamma)$, and where $j$ is fresh and $\rho' = (\rho
\deg_j, i=j)$ with $\deg_j \colon I,j \to I$ induced from the
inclusion $I \subseteq I,j$.

The interpretation of raw substitutions $\den{\Delta;\sigma}$
is a (partial) family of sets $\den{\Delta;\sigma} (I,\rho)$
for $I\in \CC$ and $\rho \in \den{\Delta} (I)$.  We set
\begin{align*}
  \den{\Delta;()} \rho &= *, & %
  \den{\Delta; (\sigma, x / t)} \rho &= %
  (\den{\Delta;\sigma} \rho, \den{\Delta;t} \rho) \quad \text{if } x
  \notin \dom(\sigma),
\end{align*}
where $*$ is the unique element of $\top (I)$.  This concludes the
definition of the interpretation of syntax.
\medskip

In the following $\alpha$ stands for either a raw term or raw
substitution.  In the latter case, $\alpha \sigma$ denotes composition
of substitutions.
\begin{lemma}
  \label{lem:face-weakening}
  Let $\Gamma'$ be like $\Gamma$ but with some $\varphi$'s inserted,
  and assume both $\den{\Gamma}$ and $\den{\Gamma'}$ are defined;
  then:
  \begin{enumerate}
  \item $\den{\Gamma'}$ is a sub-cubical set of $\den{\Gamma}$;
  \item if $\den{\Gamma;\alpha}$ is defined, then so is
    $\den{\Gamma';\alpha}$ and they agree on $\den{\Gamma'}$.
  \end{enumerate}
\end{lemma}
% \begin{proof}
%   By induction on the (syntactic) complexity of $\Gamma$ plus the
%   complexity of $\alpha$.
% \end{proof}

\begin{lemma}[Weakening]
  \label{lem:weakening}
  Let $\den{\Gamma}$ be defined.
  \begin{enumerate}
  \item If $\den{\Gamma,x:A,\Delta}$ is defined, then so is
    $\den{\Gamma,x:A,\Delta;x}$ which is moreover the projection to
    the $x$-part.\footnote{E.g., if $\Gamma$ is $y:B,z:C$, the
      projection to the $x$-part maps $(b,(c,(a,\delta)))$ to $a$, and
      the projection to the $\Gamma$-part maps $(b,(c,\delta))$ to
      $(b,c)$.}
  \item If $\den{\Gamma,\Delta}$ is defined, then so is
    $\den{\Gamma,\Delta;\id_\Gamma}$ which is moreover the projection
    to the $\Gamma$-part.
  \item If $\den{\Gamma,\Delta}$, $\den{\Gamma;\alpha}$ are defined
    and the variables in $\Delta$ are fresh for $\alpha$, then
    $\den{\Gamma,\Delta;\alpha}$ is defined and for $\rho \in
    \den{\Gamma,\Delta} (I)$:
    \[
    \den{\Gamma;\alpha} (\den{\Gamma,\Delta; \id_\Gamma} \rho) =
    \den{\Gamma,\Delta;\alpha} \rho
    \]
  \end{enumerate}
\end{lemma}

\begin{lemma}[Substitution]
  \label{lem:subst-ok}
  For $\den{\Gamma}$,$\den{\Delta}$, $\den{\Delta;\sigma}$, and
  $\den{\Gamma;\alpha}$ defined with $\dom(\Gamma) = \dom (\sigma)$
  (as lists), also $\den{\Delta; \alpha \sigma}$ is defined and for
  $\rho \in \den{\Delta}(I)$:
  \[
  \den{\Gamma;\alpha} (\den{\Delta;\sigma} \rho) = \den{\Delta; \alpha
    \sigma} \rho
  \]
  % \begin{enumerate}
  % \item if $\den{\Gamma;t}$ is defined, then also $\den{\Delta; t
  %     \sigma}$ is defined and $\den{\Gamma;t} (\den{\Delta;\sigma}
  %   \rho) = \den{\Delta; t \sigma} \rho$;
  % \item if $\den{\Gamma; \tau}$ is defined, then also $\den{\Delta;
  %     \tau \sigma}$ is defined and $\den{\Gamma;\tau}
  %   (\den{\Delta;\sigma} \rho) = \den{\Delta; \tau \sigma} \rho$ where
  %   $\tau \sigma$ denotes composition of substitutions.
  % \end{enumerate}
\end{lemma}

\begin{lemma}
  \label{lem:term-ok}
  If $\den{\Gamma}$ is defined and a cubical set, and
  $\den{\Gamma;\alpha}$ is defined, then $(\den{\Gamma;\alpha} \rho) f
  = \den{\Gamma;\alpha} (\rho f)$.
\end{lemma}

To state the next theorem let us set $\den{\Gamma;\II} = \code{\II}$
and $\den{\Gamma;\FF} = \code{\FF}$ as elements of $\Ty_0
(\den{\Gamma})$.

\begin{theorem}[Soundness]
  \label{thm:sound}
  We have the following implications, and all occurrences of $\den{-}$
  in the conclusions are defined.  In~\eqref{item:sound-ty}
  and~\eqref{item:sound-ty-eq} we allow $A$ to be $\II$ or $\FF$.
  \begin{enumerate}
  \item if $\Gamma \der {}$, then $\den{\Gamma}$ is a cubical set;
  \item if $\Gamma \der A$, then $\den{\Gamma; A} \in \Ter
    (\den{\Gamma}; \UU_1)$;
  \item\label{item:sound-ty} if $\Gamma \der t : A$, then
    $\den{\Gamma; t} \in \Ter (\den{\Gamma}; \El {\den{\Gamma;A}})$;
  \item if $\Gamma \der A = B$, then $\den{\Gamma;A} = \den
    {\Gamma;B}$;
  \item\label{item:sound-ty-eq} if $\Gamma \der a = b : A$, then
    $\den{\Gamma;a} = \den {\Gamma;b}$;
  \item if $\Gamma \der \sigma : \Delta$, then $\den {\Gamma;\sigma}$
    restricts to a natural transformation $\den{\Gamma}\to
    \den{\Delta}$.
  \end{enumerate}
\end{theorem}

\section{Extensions: identity types and higher inductive types}
\label{sec:extensions}

In this section we consider possible extensions to cubical type
theory. The first is an identity type defined using path types whose
elimination principle holds as a judgmental equality. The
second are two examples of higher inductive types.

\subsection{Identity types}\label{sec:identitytypes}
% - Nice identity type
% - UA for Id
% - Factorization?

We can use the path type to represent equalities. Using the
composition operation, we can indeed build a substitution function
$P(a) \rightarrow P(b)$ from any path between $a$ and $b$. However,
since we don't have in general the judgmental equality
$\transport^i~A~a_0 = a_0$ if $A$ is independent of $i$ (which is an
equality that we cannot expect geometrically in general, as shown in
Appendix~\ref{sec-singular}), this substitution function does not need to be the
constant function when the path is constant. This means that, as in
the previous model~\cite{BCH,simonlic}, we don't get an interpretation
of Martin-L\"of identity type~\cite{ML75} with the standard judgmental
equalities.

However, we can define another type which {\em does} give an
interpretation of this identity type following an idea of Andrew Swan.

\subsubsection*{Identity types}

The basic idea of $\Id~A~a_0~a_1$ is to define it in terms of
$\Path~A~a_0~a_1$ but also mark the paths where they are known to be
constant.  Formally, the formation and introduction rules are
\begin{mathpar}
  \inferrule{%
    \Gamma \der A \\ \Gamma \der a_0 : A \\ \Gamma \der a_1 : A
  }%
  { \Gamma \der \Id~A~a_0~a_1}%
  \and %
  \inferrule{%
    \Gamma \der \omega : \Path~A~a_0~a_1[\varphi \mapsto \pabs{i} a_0]
  }%
  {\Gamma \der (\omega,\varphi) : \Id~A~a_0~a_1}
\end{mathpar}
and we can define $\idrefl a = (\refl a, 1_\FF) : \Id~A~a~a$ for $a :
A$.  The elimination rule, given $\Gamma \der a :A$, is
\begin{mathpar}
  \inferrule {%
    %\Gamma \der a : A\\
    \Gamma, x : A, \alpha : \Id~A~a~x \der C\\
    \Gamma \der d : C ( x/ a,\alpha / \idrefl a)\\
    \Gamma \der b : A\\
    \Gamma \der \beta : \Id~A~a~b
  }%
  {\Gamma \der \J_{x,\alpha.C}~d~b~\beta : C (x/b, \alpha / \beta) }
\end{mathpar}
together with the following judgmental equality in case $\beta$ is of
the form $(\omega, \varphi)$
\[
\J~d~b~\beta = \comp^i~C(x/\omega \; i, \alpha / \beta^* (i))~[\varphi
\mapsto d]~d
\]
where $\Gamma, i : \II \der \beta^* (i) : \Id~A~a~(\omega \; i)$ is given by
\[
\beta^*(i) = (\pabs{j} \omega \; (i \land j),\varphi \lor (i=0)).
\]
% which is well defined since
% \[
% \Gamma,i:\II,(i=0) \der \pabs{j} \omega \; (i \wedge j) = \pabs{j} a%
% \text{ and }%
% \Gamma,i:\II,\varphi \der \pabs{j} \omega \; (i\wedge j) = \pabs{j} a.
% \]
Note that with this definition we get $\J~d~a~(\idrefl a) = d$ as
desired.

The composition operation for $\Id$ is explained as follows.  Given
$\Gamma, i : \II \der \Id~A~a_0~a_1$, $\Gamma, \varphi, i : \II \der
(\omega,\psi) : \Id~A~a_0~a_1$, and $\Gamma \der (\omega_0,\psi_0) :
(\Id~A~a_0~a_1) (i0) [\varphi \mapsto (\omega (i0), \psi (i0))]$ we
have the judgmental equality
\[
\comp^i~(\Id~A~a_0~a_1)~[\varphi \mapsto
(\omega,\psi)]~(\omega_0,\psi_0) = %
(\comp^i~(\Path~A~a_0~a_1)~[\varphi \mapsto \omega]~\omega_0, \varphi
\land \psi (i1)).
\]

It can then be shown that the types $\Id~A~a~b$ and $\Path~A~a~b$ are
($\Path$)-equivalent. In particular, a type is ($\Path$)-contractible
if, and only if, it is ($\Id$)-contractible. The univalence axiom,
proved in \sect{sec:univalence-axiom} for the $\Path$-type, hence
holds as well for the $\Id$-type.\footnote{This has been formally
  verified using the \Haskell{} implementation:\\
  \url{https://github.com/mortberg/cubicaltt/blob/v1.0/examples/idtypes.ctt}}

\subsubsection*{Cofibration-trivial fibration factorization}

The same idea can be used to factorize an arbitrary map of (not
necessary fibrant) cubical sets $f : A \rightarrow B$ into a
cofibration followed by a trivial fibration. We define a ``trivial
fibration'' to be a first projection from a total space of a
contractible family of types and a ``cofibration'' to be a map that
has the left lifting property against any trivial fibration. For this
we define, for $b : B$, the type $T_f(b)$ to be the type of elements
$[\varphi \mapsto a]$ with $\varphi\der a : A$ and $\varphi \der f~a =
b : B$.
\begin{theorem}
The type $T_f(b)$ is contractible and the map
\[
  A \rightarrow (b : B) \times T_f(b), %
  \qquad %
  a \longmapsto (f~a,[1_\FF \mapsto a])
\]
is a cofibration.
\end{theorem}

The definition of the identity type can be seen as a special case of
this, if we take the $B$ the type of paths in $A$ and for $f$ the
constant path function.

\subsection{Higher inductive types}\label{sec:hits}
% - circle
% - prop trunc
% - sphere
% - Data types (data,split...) and HITs (hdata,hComp,nonstandard N...)?

In this section we consider the extension of cubical type theory with
two different higher inductive types: spheres and propositional
truncation. The presentation in this section is syntactical, but it
can be directly translated into semantic definitions.

% These generalize standard inductive types by allowing not
% only constructors for points, but also constructor for paths and
% higher paths.

\subsubsection*{Extension to dependent path types}

In order to formulate the elimination rules for higher inductive
types, we need to extend the path type to {\em dependent path type},
which is described by the following rules. If $i : \II \der A$ and
$\der a_0 : A(i0),~a_1 : A(i1)$, then $\der \Path^i~A~a_0~a_1$. The
introduction rule is that $\der \pabs{i}{t} : \Path^i~A~t(i0)~t(i1)$
if $i : \II \der t : A$. The elimination rule is
$\der p \; r : A\subst{i}{r}$ if $\der p : \Path^i~A~a_0~a_1$ with
equalities $p \; 0 = a_0 : A(i0)$ and $p \; 1 = a_1 : A(i1)$.

\subsubsection*{Spheres}

We define the circle, $\Sp^1$, by the rules:
\begin{mathpar}
  \inferrule {\Gamma \der {}} {\Gamma\der \Sp^1} %
  \and %
  \inferrule {\Gamma \der {}} {\Gamma\der \base:\Sp^1} %
  \and %
  \inferrule {\Gamma \der r : \II} {\Gamma \der \LOOP(r) : \Sp^1}
\end{mathpar}
with the equalities $\LOOP(0) = \LOOP(1) = \base$.
%% %
%% \begin{mathpar}
%%   \LOOP(0) = \base %
%%   \and %
%%   \LOOP(1) = \base %
%% \end{mathpar}

%The circle hence has one point, given by the constructor $\base$, and
%a non-trivial path, $\LOOP$, whose endpoints are $\base$.

Since we want to represent the {\em free} type with one base point and a loop,
we add composition as a {\em constructor} operation $\hcomp^i$:
\begin{mathpar}
  \inferrule {\Gamma, \varphi, i : \II \der u : \Sp^1 \\
              \Gamma \der u_0 : \Sp^1[\varphi \mapsto u(i0)]}
             {\Gamma \der \hcomp^i~[\varphi \mapsto u]~u_0 : \Sp^1} %
\end{mathpar}
with the equality $\hcomp^i~[1_\FF \mapsto u]~u_0 = u(i1)$.

Given a dependent type $x : \Sp^1 \der A$ and $a : A\subst{x}{\base}$
and $l : \Path^i~A\subst{x}{\LOOP(i)}~a~a$ we can define a function
$g : (x : \Sp^1) \rightarrow A$ by the equations\footnote{For the equation
$g~\LOOP(r) = l~r$, it may be that $l$ and $r$ are dependent on the same
name $i$, and we could not have followed this definition in the framework
of \cite{BCH}.} $g~\base = a$ and
$g~\LOOP(r) = l~r$ and
\[
g~(\hcomp^i~[\varphi\mapsto u]~u_0) =
\comp^i~A\subst{x}{v}~[\varphi\mapsto g~u]~(g~u_0)
\]
where $v = \Comp^i~\Sp^1~[\varphi\mapsto u]~u_0 =
\hcomp^j~[\varphi\mapsto u\subst{i}{i\wedge j}, (i = 0) \mapsto u_0]~u_0$.

This definition is non ambiguous since
$l~0 = l~1 = a$.

\medskip

We have a similar definition for $\Sp^n$ taking as constructors
$\base$ and $\LOOP(r_1,\dots,r_n)$.

\subsubsection*{Propositional truncation}

We define the propositional truncation, $\inh~A$, of a type $A$ by the
rules:
\begin{mathpar}
  \inferrule {\Gamma \der A} {\Gamma \der \inh~A} %
  \and %
  \inferrule {\Gamma \der a : A} {\Gamma \der \inc~a : \inh~A} %
  \and %
  \inferrule {\Gamma \der u_0 : \inh~A \\ \Gamma \der u_1 : \inh~A \\
              \Gamma \der r : \II}
             {\Gamma \der \squash(u_0,u_1,r) : \inh~A} %
\end{mathpar}
with the equalities $\squash(u_0,u_1,0) = u_0$ and $\squash(u_0,u_1,1) = u_1$.
%% %
%% \begin{mathpar}
%%   \squash(u_0,u_1,0) = u_0 %
%%   \and %
%%   \squash(u_0,u_1,1) = u_1 %
%% \end{mathpar}

As before, we add composition as a constructor, but only in the
form\footnote{This restriction on the constructor is essential for the
  justification of the elimination rule below.}
\begin{mathpar}
  \inferrule {\Gamma,\varphi,i : \II \der u : \inh~A \\
              \Gamma \der u_0 : \inh~A[\varphi \mapsto u(i0)]}
             {\Gamma \der \hcomp^i~[\varphi \mapsto u]~u_0 : \inh~A} %
\end{mathpar}
with the equality $\hcomp^i~[1_\FF \mapsto u]~u_0 = u(i1)$.

This provides only a definition of
$\comp^i~(\inh~A)~[\varphi\mapsto u]~u_0$ in the case where $A$ is
independent of $i$, and we have to explain how to define the general
case.

\medskip

In order to do this, we define first two operations
\begin{mathpar}
  \inferrule {\Gamma,i : \II \der A \\ \Gamma \der u_0 : \inh~A(i0)}
             {\Gamma \der \transp~u_0:\inh~A(i1)} %
  \and %
  \inferrule {\Gamma,i : \II \der A \\ \Gamma,i : \II \der u : \inh~A}
             {\Gamma \der \squeeze^i~u:\Path~(\inh~A(i1))~(\transp~u(i0))~u(i1)} %
\end{mathpar}
by the equations
\[
\begin{array}{lcl}
\transp~(\inc~a) & = & \inc~(\comp^i~A~[]~a) \\
\transp~(\squash(u_0,u_1,r)) & = & \squash(\transp~u_0,\transp~u_1,r) \\
\transp~(\hcomp^j~[\varphi\mapsto u]~u_0) & = & \hcomp^j~[\varphi\mapsto \transp~u]~(\transp~u_0)\\
 & & \\
\squeeze^i~(\inc~a) & = & \pabs{i}{\inc~(\comp^j~A(i\vee j)~[(i=1)\mapsto a(i1)]~a)} \\
\squeeze^i~(\squash(u_0,u_1,r)) & = & \pabs{k}{\squash(\squeeze^i~u_0~k,\squeeze^i~u_1~k,r\subst{i}{k})} \\
\squeeze^i~(\hcomp^j~[\varphi\mapsto u]~v) & = & \pabs{k}{\hcomp^j~S~(\squeeze^i~v~k)}
\end{array}
\]
where $S$ is the system
\[
[ \delta \mapsto \squeeze^i~u~k,~
  \varphi\subst{i}{k} \wedge (k=0) \mapsto \transp~u(i0),~
  \varphi\subst{i}{k} \wedge (k=1) \mapsto u(i1)]
\]
% \[\begin{array}{lcl}
% \delta & \mapsto & \squeeze^i~u~k,\\
% \varphi\subst{i}{k} \wedge (k=0) & \mapsto & \transp~u(i0),\\
% \varphi\subst{i}{k} \wedge (k=1) & \mapsto & u(i1)
% \end{array} \]
and $\delta = \forall i. \varphi$, using Lemma \ref{decomp}.

Using these operations, we can define the general composition
\begin{mathpar}
  \inferrule {\Gamma,i:\II \der A \\
              \Gamma,\varphi,i:\II \der u : \inh~A \\
              \Gamma \der u_0 : \inh~A(i0)[\varphi \mapsto u(i0)]}
             {\Gamma \der \comp^i~(\inh~A)~[\varphi \mapsto u]~u_0 :
                \inh~A(i1)[\varphi \mapsto u(i1)]} %
\end{mathpar}
by
$ \Gamma \der \comp^i~(\inh~A)~[\varphi \mapsto u]~u_0 =
\hcomp^j~[\varphi \mapsto \squeeze^i~u~j]~(\transp~u_0) : \inh~A(i1)$.

\medskip

Given $\Gamma \der B$ and
$\Gamma \der q : (x~y : B) \rightarrow \Path~B~x~y$ and
$f : A \rightarrow B$ we define $g : \inh~A \rightarrow B$ by the
equations
\[
\begin{array}{lcl}
g~(\inc~a) & = & f~a \\
g~(\squash(u_0,u_1,r)) & = & q~(g~u_0)~(g~u_1)~r \\
g~(\hcomp^j~[\varphi \mapsto u]~u_0) & = & \comp^j~B~[\varphi \mapsto g~u]~(g~u_0)
\end{array}
\]

\section{Related and future work}\label{section:conclusion}

Cubical ideas have proved useful to reason about equality in homotopy
type theory~\cite{LicataBrunerie}. In cubical type theory these
techniques could be simplified as there are new judgmental equalities
and better notations for manipulating higher dimensional cubes. Indeed
some simple experiments using the \Haskell{} implementation have shown
that we can simplify some constructions in synthetic homotopy
theory.\footnote{For details see:
  \url{https://github.com/mortberg/cubicaltt/tree/master/examples/}}

Other approaches to extending intensional type theory with
extensionality principles can be found
in~\cite{Altenkirch99,Polonsky14}. These approaches have close
connections to techniques for internalizing parametricity in type
theory~\cite{ttincolor}. Further, nominal extensions to
$\lambda$-calculus and semantical ideas related to the ones presented
in this paper have recently also proved useful for justifying type
theory with internalized parametricity~\cite{paramtt}.

The paper~\cite{GambinoSattler} provides a general framework for
analyzing the uniformity condition, which applies to simplicial and
cubical sets.

Large parts of the semantics presented in this paper have been
formally verified in NuPrl by Mark Bickford\footnote{For details see:
  \url{http://www.nuprl.org/wip/Mathematics/cubical!type!theory/}}, in
particular, the definition of Kan filling in terms of composition as
in \sect{sec:kan-filling} and composition for glueing as given in
\sect{sec:composition-glueing}.

\medskip

Following the usual reducibility method, we expect it to be possible
to adapt our presheaf semantics to a proof of normalization and
decidability of type checking.  A first step in this direction is the
proof of canonicity in~\cite{Huber16}.  We end the paper with a list
of open problems and conjectures:

\begin{enumerate}

%\item Prove decidability of type-checking and conversion? (hopefully
%  we have an answer by end of November)

%\item Show that $\II$ does not have a uniform composition operation

% \item Show that any cubical group has a uniform Kan composition
%   operation.

\item Extend the semantics of identity types to the semantics of
  inductive families.

\item Give a general syntax and semantics of higher inductive types.

\item Extend the system with resizing rules and show normalization.

\item Is there a model where $\Path$ and $\Id$ coincide?

\end{enumerate}

%% Local Variables:
%% ispell-local-dictionary: "english"
%% mode: latex
%% TeX-master: "main"
%% End:

\subparagraph*{Acknowledgements}

This work originates from discussions between the four authors around an implementation
of a type system corresponding to the model described in~\cite{BCH}. This implementation
indicated a problem with the representation of higher inductive types, e.g., the elimination rule
for the circle, and suggested the need of extending this cubical model with a diagonal 
operation. The general framework (uniformity condition, connections, semantics of spheres and propositional 
truncation) is due to the second author. In particular, the glueing operation with its composition was 
introduced as a generalization of the operation described in~\cite{BCH} transforming an equivalence
into a path, and with the condition $A = \Glue~[]~A$. In a first attempt, we tried to force ``regularity'', 
i.e., the equation $\transport~i~A~a_0 = a_0$ if $A$ is independent of $i$ (which seemed to be necessary
in  order to get filling from compositions, and which implies $\Path = \Id$).
There was a problem however for getting regularity for the universe, that was discovered by Dan Licata (from 
discussions with Carlo Angiuli and Bob Harper). Thanks to this discovery, it was realized that 
regularity is actually not needed for the model to work. In particular, the second author
adapted the definition of filling from composition as in Section \ref{sec:kan-filling},
the third author noticed that we can remove the condition $A = \Glue~[]~A$, and together with the last author,
they derived the univalence axiom from the glueing operation as presented in the appendix.
This was surprising since glueing was introduced
a priori only as a way to transform equivalences into paths, but was later explained by a remark of Dan Licata
(also presented in the appendix: we get univalence as soon as the transport map associated to this
path is path equal to the given equivalence). 
The second author introduced then the restriction operation $\Gamma,\varphi$
on contexts, which, as noticed by Christian Sattler, can be seen as an explicit syntax for the notion of 
cofibration, and designed the other proof of univalence in \sect{sec:univalence-axiom} from
discussions between Nicola Gambino, Peter LeFanu Lumsdaine and the third author. Not having regularity, the type of 
paths is not the same as the $\Id$ type but, as explained in \sect{sec:identitytypes}, we can recover
the usual identity type from the path type, following an idea of Andrew Swan.

\medskip

The authors would like to thank the referees and Mart\'in Escard\'o, Georges Gonthier,
Dan Grayson, Peter Hancock, Dan Licata, Peter LeFanu Lumsdaine,
Christian Sattler, Andrew Swan, Vladimir Voevodsky for many
interesting discussions and remarks.

\bibliographystyle{plainurl}% the recommended bibstyle
\bibliography{references}

\newpage

\appendix
\section{Details of composition for
  glueing}\label{appendix:composition-glueing}

We build the element~$\Gamma \der b_1 = \comp^i~B~[\psi\mapsto
b]~b_0:(\Glue~[\varphi\mapsto(T,\myeq{})]~A)(i1)$ as the element
$\glue~[\varphi(i1) \mapsto t_1]~a_1$ where
\begin{align*}
  \Gamma,\varphi(i1) \der & \; t_1 : T(i1)[\psi \mapsto b(i1)]\\
  \Gamma \der & \; a_1 : A(i1)[\varphi(i1) \mapsto \myeq{}(i1) \; t_1,\psi \mapsto (\ugl~b)(i1)]
\end{align*}

As intermediate steps, we gradually build elements that satisfy more
and more of the equations that the final elements $t_1$ and $a_1$
should satisfy. The construction of these is given in five steps.

Before explaining how we can define them and why they are well
defined, we illustrate the construction in \fig{fig:compglue},
with $\psi = (j=1)$ and $\varphi = (i=0) \vee (j=1)\vee (i=1)$.

We pose $\delta = \forall i.\varphi$ (cf.\ \sect{sec:pathtypes}), so
that we have that~$\delta$ is independent from~$i$, and in our example
$\delta = (j = 1)$ and it represents the right-hand side of the
picture.
\begin{enumerate}
\item The element $a'_1:A(i1)$ is a first approximation of~$a_1$, but
  $a'_1$ is not necessarily in the image of~$\myeq{}(i1)$
  in~$\Gamma,\varphi(i1)$;
\item the partial element $\inctxt{\delta}{t'_1:T(i1)}$, which is a
  partial final result for~$\inctxt{\varphi(i1)}{t_1}$;
\item the partial path $\inctxt{\delta}{\omega}$, between~$a'_1$ and
  the image of~$t'_1$;
\item both the final element $\inctxt{\varphi(i1)}{t_1}$ and a path
  $\inctxt{\varphi(i1)}{\alpha}$ between~$a'_1$ and
  $\myeq{}(i1) \; t_1$;
\item finally, we build $a_1$ from~$a'_1$ and~$\alpha$.
\end{enumerate}
\begin{figure}[!h]
  \centering
  $$
\begin{tikzpicture}
  \coordinate (A00) at (0,0);
  \coordinate (A10) at ++ (3,0) ;
  \coordinate (A'11) at ++ (3,2) ;
  \coordinate (A'01) at ++ (0,2) ;
  \coordinate (A01) at ($(A'01) + (0,1.5)$);
  \coordinate (A11) at ($(A'11) + (0,1.5)$);

  \coordinate (O) at ($(A00) + (-4,1)$);
  \draw[fun line] (O) -- node[left] () {$i$} ++(0,1);
  \draw[fun line] (O) -- node[below] () {$j$} ++(1,0);

  % inner
  \begin{pgfonlayer}{background}
    \draw [path line] (A00) -- node [above] (a_0) {$\ugl~b_0$}
    (A10) -- coordinate (aj1)
    (A'11) -- coordinate (midA') (A'01);
    \draw [path line] (A01) -- coordinate (midA)
    node[above] () {Step 5: $a_1$} (A11);
    \draw (midA') node[below] {Step 1: $a'_1$} (A'01);
  \end{pgfonlayer}
  \draw (aj1) node[right, fill=white] () {$(\ugl~b)(j1)$};

  % links
  \begin{pgfonlayer}{background}
    \draw [fun line, <-] (A00) --
    node[left, xshift=-5] (f00) {$f$} ++ (-1,-1) coordinate (T00);
    \draw [fun line, <-] (A01) --
    node[left, xshift=-5] (f00) {$f$} ++ (-1,1) coordinate (T01);
    \draw [fun line, <-] (A10) --
    node[left, xshift=-5] (f00) {$f$} ++ (1,-1) coordinate (T10);
    \draw [fun line, <-] (A11) --
    node[left, xshift=-5] (f00) {$f$} ++ (1,1) coordinate (T11);
  \end{pgfonlayer}

  \fill (T11) circle [radius=2pt] node [right,xshift=10] () {Step 2:
    $\inctxt{\delta}{t'_1}$};

  \fill (A11) circle [radius=2pt] node [right,xshift=10,fill=white] ()
  {$\inctxt{\delta}{\myeq{}(i1) \; t'_1}$};

  % outer
  \begin{pgfonlayer}{background}
    \draw [path line] (T00) -- node[below] (b_0)
    {$\inctxt{\varphi(i0)}{b_0}$}
    (T10) -- coordinate (bj1)
    (T11) -- node[above] (t_1) {Step 4: $\inctxt{\varphi(i1)}{t_1}$} (T01);
    \draw [fun line] (b_0) -- node[left] {$\myeq{}$} (a_0);
  \end{pgfonlayer}
  \draw (bj1) node[right, fill=white] () {$b(j1)$};

  \draw [comp line] (A'11) -- node[right, xshift=10,fill=white] ()
  {Step 3: $\inctxt{\delta}{\omega}$ constant on $\psi$} (A11);

  % inner top diag edge
  % \coordinate (A''01) at ($(A'01) !0.25! (A01)$);
  % \draw [path line] (A''01) -- coordinate (midA'')
  %   node[below,yshift=-5] () {Step 4: $a''_1$} (A11);
  \draw [comp line] (midA') --
  node[left] () {Step 4': $\inctxt{\varphi(i1)}{\alpha}$} (midA);

\end{tikzpicture}
$$
\caption{Composition for glueing}
\label{fig:compglue}
\end{figure}

We define:
\begin{align*}
%  \Gamma,i:\II,\psi,\varphi\der t : T\label{eq:glue-t}\\
%  \Gamma,i:\II,\psi\der a = & \; \ugl~ b: A[\varphi\mapsto \myeq{} \;  b]%\label{eq:glue-a}
   \Gamma,\psi,i:\II \der a = & \; \ugl~ b: A[\varphi\mapsto \myeq{} \;  b]%\label{eq:glue-a}
  \\
%  \Gamma,\varphi(i0)\der t_0 : T(i0)[\psi\mapsto t(i0)]\label{eq:glue-t0}\\
  \Gamma\der a_0 = & \; \ugl~ b_0: A(i0)[\varphi(i0)\mapsto \myeq{}(i0) \; b_0,\psi\mapsto a(i0)]%\label{eq:glue-a0}
\end{align*}

\subparagraph*{Step 1} We define~$a'_1$ as the composition of~$a$
and~$\ugl~b_0$, in the direction~$i$, which is well defined since
$\ugl~b_0 = (\ugl~b)(i0)$ over the extent $\psi$
\begin{equation}
  \label{eq:glue-a'1}
  \Gamma \der a_1' = \comp^i~A~[\psi\mapsto~a]~a_0:A(i1)[\psi\mapsto~a(i1)]
\end{equation}
%\[\textrm{which is well defined because}\quad
%\begin{cases}
%  \Gamma,i:\II,\psi\der a : A \hfill \textrm{by~(\ref{eq:glue-a})}\\
%  \Gamma\der a_0 : A(i0)[\psi\mapsto a(i0)] \hspace{1cm} \textrm{by~(\ref{eq:glue-a0})}
%\end{cases}\]

\subparagraph*{Step 2} We also define~$t'_1$ as the composition of~$b$
and~$b_0$, in the direction~$i$:
\begin{equation}
  \label{eq:glue-t'1}
  \Gamma,\delta\der t_1' = \comp^i~T~[\psi\mapsto~b]~b_0:T(i1)[\psi\mapsto~b(i1)]
\end{equation}
\[\textrm{which is well defined because}\quad
\begin{cases}
  \Gamma,\delta,i:\II,\psi\der b : T \hfill
   \textrm{by Lemma \ref{admissible}} \\
  \Gamma,\delta\der b_0 : T(i0)[\psi\mapsto b(i0)] \hspace{1cm} \textrm{as~%(\ref{eq:glue-t0}) and
    $\delta \leqslant \varphi(i0)$}
\end{cases}\]
Moreover, since
\[
\begin{cases}
  \Gamma,\delta,\psi,i:\II \der a = \myeq{} \; b \hspace{1cm} \textrm{by~%(\ref{eq:glue-a}) and
    $\delta \leqslant \varphi$ ~~~~} \\
  \Gamma,\delta\der~a_0=\myeq{}(i0)~b_0 \hfill \textrm{by~%(\ref{eq:glue-a0}) and
    $\delta \leqslant \varphi(i0)$}
\end{cases}\]
we can re-express~$a'_1$ on the extent $\delta$
\[
\Gamma, \delta \der a_1' = \comp^i~A~[\psi\mapsto\myeq{}~b]~\left(\myeq{}(i0)~b_0\right)
\]
%
%Hence, the element~$a'_1$ restricted to the extent $\delta$ is the composition of images by~$\myeq{}$.

\subparagraph*{Step 3} We can hence find a path~$\omega$
connecting~$a'_1$ and~$\myeq{}(i1) \; t'_1$ in~$\Gamma,\delta$
using Lemma~\ref{pres}:
\[
\Gamma,\delta \der \omega = \pres{i}~\myeq{}~[\psi\mapsto b]~b_0 :
\left(\Path~A(i1)~a_1'~\left(\myeq{}(i1) \;
    t'_1\right)\right)[\psi\mapsto\pabs{j} a(i1)]
\]
Picking a fresh name~$j$, we have
\begin{equation}
  \Gamma,\delta,j:\II\der \omega~j : A(i1)[(j=0)\mapsto
  a'_1,(j=1)\mapsto \myeq{}(i1) \; t'_1,\psi\mapsto a(i1)]\label{eq:glue-omega}
\end{equation}

% \subparagraph*{Step 4} We then define~$a''_1$ as the composition
% of~$[\delta \mapsto \omega \; j, \psi \mapsto a(i1)]$ and~$a'_1$, in
% the direction~$j$:
% %
% \[\Gamma\der a''_1 = \comp^j~A(i1)~[\delta\mapsto \omega \; j,~\psi\mapsto
% a(i1)]~a_1' : A(i1)[\delta\mapsto\omega \; 1, \psi\mapsto a(i1)]
% \]
% \[\textrm{which is well defined because}\quad
% \begin{cases}
%   \Gamma,j:\II,\delta,\psi\der \omega~j = a(i1)
%   \hfill \textrm{by~(\ref{eq:glue-omega})} \\
%   \Gamma,j:\II,\delta\der \omega~0 = a'_1
%   \hfill \textrm{by~(\ref{eq:glue-omega})} \\
%   \Gamma,j:\II,\psi\der a(i1) = a'_1
%   \hspace{1cm} \textrm{by~(\ref{eq:glue-a'1})}
% \end{cases}
% \]
% and since $\Gamma,\delta\der\omega~1 = \myeq{}(i1) \; t'_1$,
% because~$\omega$ is a partial path between $a'_1$ and
% $\myeq{}(i1) \; t'_1$, we can re-express the type of~$a''_1$ in the
% following way:
% %
% \begin{equation}
% \Gamma\der a''_1 : A(i1)[\delta\mapsto\myeq{}(i1) \; t'_1,~\psi\mapsto a(i1)]\label{eq:glue-a''1}
% \end{equation}

\subparagraph*{Step 4} Now we define the final element~$t_1$ as the
inverse image of~$a'_1$ by~$\myeq{}(i1)$, together with the
path~$\alpha$ between~$a'_1$ and~$\myeq{}(i1) \; t_1$, in
$\Gamma,\varphi(i1)\der$, using Lemma~\ref{equiv}:
\[\Gamma,\varphi(i1)\der(t_1,\alpha) =
\eq~\myeq{}(i1)~[\delta\mapsto~(t'_1,\omega),\psi\mapsto~\left(b(i1),\pabs{j}a'_1\right)]~a_1'\]
\[{\textrm{with}\quad}\begin{cases}
  \Gamma,\varphi(i1)\der t_1 :
  T(i1)[\delta\mapsto~t'_1,\psi\mapsto~b(i1)] % \label{eq:glue-t1}
  \\
  \Gamma,\varphi(i1)\der\alpha :
  \left(\Path~A(i1)~a'_1~\left(\myeq{}(i1) \; t_1\right)\right)
  [\delta\mapsto\omega,\psi\mapsto~\pabs{j} a'_1)]
\end{cases}
\]
These are well defined because the two systems in~$\delta$ and~$\psi$
are compatible:
\[\begin{cases}
  \Gamma, \delta, \psi \der t'_1 = b(i1) & \textrm{by~(\ref{eq:glue-t'1})} \\
  \Gamma, \delta, \psi \der \omega = \pabs{j}a'_1 &
  \textrm{by~(\ref{eq:glue-omega}) and~(\ref{eq:glue-a'1})}
\end{cases}
\]
% %
% \[\begin{cases}
%   \Gamma,\varphi(i1),\delta\der a'_1 = \myeq{}(i1) \; t'_1
%   \hfill \textrm{by~(\ref{eq:delta-a'1})} \\
%   \Gamma,\varphi(i1),\psi\der a'_1 = a(i1) = \myeq{}(i1) \; b(i1)
%   \hspace{1cm} \textrm{by~(\ref{eq:glue-a'1}) in~$\Gamma,\psi\der$} %and~(\ref{eq:glue-a}) in~$\Gamma,\varphi\der$}
% \end{cases}\]
%
Picking a fresh name~$j$, we have
\begin{equation}
\Gamma,\varphi(i1),j:\II\der \alpha~j : A(i1)[(j=0)\mapsto
  a'_1,~(j=1)\mapsto \myeq{}(i1) \; t_1,~\delta\mapsto~a'_1,~\psi\mapsto a(i1)]\label{eq:glue-alpha}
\end{equation}

\subparagraph*{Step 5} Finally, we define~$a_1$ by composition
of~$\alpha$ and~$a'_1$:
\[\Gamma\der~a_1:=\comp^j~A(i1)~[\varphi(i1)\mapsto\alpha~j,\psi\mapsto~a(i1)]~a'_1:
A(i1)[\varphi(i1)\mapsto\alpha~1,\psi\mapsto~a(i1)]\]
\[\textrm{which is well defined because}\quad
\begin{cases}
  \Gamma,j:\II,\varphi(i1),\psi \der \alpha~j = a(i1)
  &\textrm{by~(\ref{eq:glue-alpha})} \\
  \Gamma,\varphi(i1)\der \alpha~0 = a'_1
  &\textrm{by~(\ref{eq:glue-alpha})} \\
  \Gamma,\psi\der a(i1) = a'_1
  &\textrm{by~(\ref{eq:glue-a'1})}
\end{cases}
\]
and since $\Gamma,\varphi(i1)\der\alpha~1 = \myeq{}(i1) \; t_1$, we can
re-express the type of~$a_1$ in the following way:
\[
  \Gamma\der a_1 : A(i1)[\varphi(i1)\mapsto\myeq{}(i1) \;
  t_1,~\psi\mapsto a(i1)]
\]
Which is exactly what we needed to build
$\Gamma\der b_1 := \glue ~[\varphi(i1)\mapsto t_1]~a_1:
B(i1)[\psi\mapsto b(i1)]$.

\medskip

\noindent Finally we check that $b_1 = \comp^i~T~ [\psi\mapsto~b]~b_0$
on $\delta$:
\begin{align*}
  b_1  &= \glue ~[\varphi(i1)\mapsto t_1]~a_1 &&\textrm{by def.} \\
  &=t_1 : T(i1)[\delta\mapsto~t'_1,\psi\mapsto~b(i1)]
  && \textrm{as $\varphi(i1) = 1_\FF$} \\
  &=t'_1 && \textrm{as $\delta = 1_\FF$} \\
  &= \comp^i~T~[\psi\mapsto~b]~b_0 &&  \textrm{by def.}
\end{align*}
% \[\begin{array}{llll}
% b_1  & = & \glue ~[\varphi(i1)\mapsto t_1]~a_1 & \hspace{2cm} \textrm{by def.} \\
%      & = & t_1 : T(i1)[\delta\mapsto~t'_1,\psi\mapsto~t(i1)] &
%      \hspace{2cm} \textrm{as $\varphi(i1) = 1_\FF$} \\
%      & = & t'_1 & \hspace{2cm} \textrm{as $\delta = 1_\FF$} \\
%      & = & \comp^i~T~[\psi\mapsto~b]~b_0 & \hspace{2cm} \textrm{by def.}
% \end{array}\]

\section{Univalence from glueing}
\label{sec:univ-from-glue}

We also give two alternative proofs of the univalence axiom for
$\Path$ only involving the glue construction.%
\footnote{The proofs of the univalence axiom have all been formally
  verified inside the system using the \Haskell{} implementation. We
  note that the proof of Theorem~\ref{thm:unglueequiv} can be given
  such that it extends $\myeq{}.2$ and hence in
  Corollary~\ref{cor:equiv-contr} we do not need the fact that
  $\isEquiv~X~A~\myeq{}.1$ is a proposition. For details see:\\
  \url{https://github.com/mortberg/cubicaltt/blob/v1.0/examples/univalence.ctt}}
The first is a direct proof of the standard formulation of the
univalence axiom while the second goes through an alternative
formulation as in Corollary~\ref{cor:equiv-contr}.%
\footnote{The second of these proofs is inspired by a proof by Dan Licata from:\\
  \url{https://groups.google.com/d/msg/homotopytypetheory/j2KBIvDw53s/YTDK4D0NFQAJ}}

\begin{lemma}
  \label{lem:ua-prel}
  For $\Gamma \vdash A : \U$, $\Gamma \vdash B : \U$, and an
  equivalence $\Gamma \vdash \myeq : \Equiv~A~B$ we have the following
  constructions:
  \begin{enumerate}
  \item\label{item:uap1} $\Gamma \vdash \eqToPath \myeq :
    \Path\, \U\, A\, B$;
  \item\label{item:uap2} $\Gamma \vdash \Path \, (A \to B) \,
    (\transport^i (\eqToPath \myeq \, i))) \, \myeq.1$ is
    inhabited; and
  \item\label{item:uap3} if $\myeq = \eq^i (P \, i)$ for $\Gamma
    \vdash P : \Path \, \U\, A \, B$, then the following type is
    inhabited:
    \[
    \Gamma \vdash %
    \Path \, (\Path \, \U \, A \, B) \, (\eqToPath {(\eq^i (P \, i))})
    \, P
    \]
  \end{enumerate}
\end{lemma}
\begin{proof}
  For \eqref{item:uap1} we define
  \begin{align}
    \label{eq:def-eqToPath}
    \eqToPath \myeq = \pabs i {\Glue \, [ (i=0) \mapsto
      (A,\myeq), (i=1) \mapsto (B, \eq^k B)] \, B }.
  \end{align}
  Note that here $\eq^k B$ is an equivalence between $B$ and $B$ (see
  \sect{subsec:compU}). For \eqref{item:uap2} we have to closely look
  at how the composition was defined for $\Glue$. By unfolding the
  definition, we see that the left-hand side of the equality is equal
  $\myeq.1$ composed with multiple transports in a constant type;
  using filling and functional extensionality, these transports can be
  shown to be equal to the identity; for details see the formal proof.

  The term for \eqref{item:uap3} is given by:
  \begin{align*}
    \pabs {j} \pabs {i} \Glue \,[ %
    &(i=0) \mapsto (A, \eq^k (P\, k)),\\
    &(i=1) \mapsto (B, \eq^k B),\\
    &(j=1) \mapsto (P \, i, \eq^k (P (i \lor k)))]\\
    &B \qedhere
  \end{align*}
\end{proof}

\begin{corollary}[Univalence axiom]
  \label{cor:ua}
  For the canonical map
  \[
  \can : (A \, B : \U) \to \Path \,\U\,A\,B \to \Equiv~A~B
  \]
  we have that $\can \, A \, B$ is an equivalence for
  all $A : \U$ and $B : \U$.
\end{corollary}
\begin{proof}[Proof 1]
  Let us first show that the canonical map $\can$ is path equal to:
  \[
  \eq = \lambda A \, B : \U. \, \lambda P : \Path\,\U\,A\,B. \, \eq^i (P\, i)
  \]
  By function extensionality, it suffices to check this pointwise.
  Using path-induction, we may assume that $P$ is reflexivity.  In
  this case $\can\,A\,A\,\refl{A}$ is the identity equivalence by
  definition.  Because being an equivalence is a proposition, it thus
  suffices that the first component of $\eq^i A$ is propositionally
  equal to the identity.  By definition, this first component is given
  by transport (now in the constant type $A$) which is easily seen to
  be the identity using filling (see \sect{sec:kan-filling}).

  Thus it suffices to prove that $\eq\,A\,B$ is an equivalence.  To do
  so it is enough to give an inverse (see Theorems 4.2.3 and 4.2.6 of
  \cite{hott-book}). But $\eqToPath$ is a left inverse by
  Lemma~\ref{lem:ua-prel}~\eqref{item:uap3}, and a right inverse by
  Lemma~\ref{lem:ua-prel}~\eqref{item:uap2} using that being an
  equivalence is a proposition.
\end{proof}
\begin{proof}[Proof 2]
  Points~\eqref{item:uap1} and~\eqref{item:uap2} of
  Lemma~\ref{lem:ua-prel} imply that~$\Equiv~A~B$ is a retract
  of~$\Path \,\U\,A\,B$. Hence $(X : \U) \times \Equiv~A~X$ is a
  retract of $(X : \U) \times \Path\,\U\,A\,X$.  But $(X : \U) \times
  \Path\,\U\,A\,X$ is contractible, so $(X : \U) \times \Equiv~A~X$ is
  also contractible as a retract of a contractible type. As
  discussed in~\sect{sec:univalence-axiom} this is an alternative
  formulation of the univalence axiom and the rest of this proof
  follows as there.
\end{proof}

Note that the first proof uses all three of the points of
Lemma~\ref{lem:ua-prel} while the second proof only uses the first
two. As the second proof only uses the first two points it is possible
to prove it if point~\eqref{item:uap1} is defined as in
Example~\ref{exa:equivtopath} leading to a slightly simpler proof of
point~\eqref{item:uap2}.

\section{Singular cubical sets}\label{sec-singular}

Recall the functor $\CC \to \Top, I \mapsto [0,1]^I$ given
at~\eqref{eq:georeal} in \sect{sec:cubecat}. In particular, the face
maps $(ib) \colon I-i \to I$ (for $b=0_\II$ or $1_\II$) induce the
maps $(ib)\colon [0,1]^{I-i} \to [0, 1]^{I}$ by $i(ib)u = b$ and
$j(ib)u = ju$ if $j\neq i$ is in $I$.  If $\psi$ is in $\FF(I)$ and
$u$ in $[0,1]^I$, then $\psi u$ is a truth value.

We assume given a family of idempotent functions
$r_I:[0,1]^I\times [0,1]\rightarrow [0,1]^I\times [0,1]$ such that
\begin{enumerate}
\item $r_I(u,z) = (u,z)$ if{f} $\partial_I u = 1$ or $z = 0$, and
\item for any {\em strict} $f$ in ${\sf Hom}(I,J)$ we have
  $r_J(f\times\id)r_I = r_J(f\times\id)$.
\end{enumerate}

\medskip

Such a family can for instance be defined as in the following picture
(``retraction from above center''). If the center has coordinate
$(1/2,2)$, then $r_I(u,z) = r_I(u',z')$ is equivalent to $(2-z')
(-1+2u) = (2-z) (-1+2u')$.
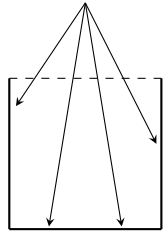
$$
\begin{tikzpicture}[scale=2]
  % bottom
  \draw [thick,solid] (0,0) --
  node [near start,below=-1pt] (DL) {}
  node [near end,below=-1pt] (DR) {}
  (1,0);
  % left and right
  \draw [thick,solid] (0,0) --
  node [near end,](L) {} (0,1);
  \draw [thick,solid] (1,0) -- node (R) {} (1,1);
  % top
  \draw [dashed] (0,1) -- (1,1);
  % rays
  \def\raypt{(0.5,1.5)}
  \draw[->,>=stealth]\raypt -- (DL);
  \draw[->,>=stealth]\raypt -- (DR);
  \draw[->,>=stealth]\raypt -- (L);
  \draw[->,>=stealth]\raypt -- (R);
\end{tikzpicture}
$$

Property (1) holds for the retraction defined by this picture.
The property (2) can be reformulated as $r_I(u,z) =
r_I(u',z')\rightarrow r_J(f u,z) = r_J(f u',z')$. It also holds
in this case,
since $r_I(u,z) = r_I(u',z')$ is then equivalent to $(2-z') (-1+2u) = (2-z) (-1+2u')$,
which implies $(2-z') (-1+2f u) = (2-z)(-1+2 f u')$ if $f$ is strict.

Using this family, we can define for each $\psi$ in $\FF(I)$ an
idempotent function
\[
r_{\psi}:[0,1]^I\times [0,1]\rightarrow [0,1]^I\times [0,1]
\]
having for fixed-points the element $(u,z)$ such that $\psi u = 1$ or $z = 0$.
This function $r_{\psi}$ is completely characterized by the following properties:
\begin{enumerate}
\item $r_{\psi} = \id$ if $\psi = 1$
\item $r_{\psi} = r_{\psi}r_I$ if $\psi \neq 1$
\item $r_{\psi}(u,z) = (u,z)$ if $z = 0$
\item $r_{\psi}((ib)\times\id) = ((ib)\times\id)r_{\psi(ib)}$
\end{enumerate}

These properties imply for instance $r_{\partial_I}(u,z) = (u,z)$ if
$\partial_I u = 1$ or $z=0$ and so they imply $r_{\partial_I} = r_I$.
They also imply that $r_{\psi}(u,z) = (u,z)$ if $\psi u = 1$.

\medskip

From these properties, we can prove the uniformity of the family of
functions $r_{\psi}$.

\begin{theorem}
  If $f$ is in ${\sf Hom}(I,J)$ and $\psi$ is in $\FF(J)$, then
  $r_{\psi}(f\times\id) = (f\times\id)r_{\psi f}$.
\end{theorem}

This is proved  by induction on the number of element of $I$ (the
  result being clear if $I$ is empty).

A particular case is $r_J(f\times\id) = (f\times\id)r_{\partial_J
  f}$. Note that, in general, $\partial_Jf$ is not $\partial_I$.

\medskip

A direct consequence of the previous theorem is the following.

\begin{corollary} The singular cubical set associated to a topological space
has a composition structure.
\end{corollary}

%% Local Variables:
%% mode: latex
%% TeX-master: "main"
%% End:

\end{document}